\newif\iflong
\newif\ifshort
\longtrue

\iflong
\else
\shorttrue
\fi

\documentclass[11pt,a4,USenglish]{article}

\usepackage{microtype}

\usepackage{fullpage}
\usepackage{authblk}
\usepackage[pagebackref]{hyperref}
\usepackage{amsthm}

\usepackage{amsmath}
\usepackage{amssymb}
\usepackage{comment}
\usepackage{amsfonts}

\usepackage{paralist}

\usepackage{cleveref}

\usepackage{tabularx}
\usepackage{tikz}
\definecolor{lipyel}{rgb}{0.99,0.78,0.07}

\newcommand{\prob}[1]{\textsc{#1}}
\newcommand{\pprob}[2]{\prob{#1}$(#2)$}

\newcommand{\coNP}{\textnormal{coNP}}
\newcommand{\coNPslashpoly}{\textnormal{\coNP{}$/$poly}}
\newcommand{\NP}{\textnormal{NP}}
\newcommand{\PTIME}{\textnormal{P}}

\newcommand{\FPT}{\textnormal{FPT}}
\newcommand{\W}[1]{\textnormal{W[#1]}}

\newcommand{\PnotNP}{$\PTIME{}\neq\NP{}$}
\newcommand{\PeqNP}{$\PTIME{}=\NP{}$}
\newcommand{\NPnotincoNPslashpoly}{$\NP{}\nsubseteq\coNPslashpoly{}$}
\newcommand{\NPincoNPslashpoly}{$\NP{}\subseteq\coNPslashpoly{}$}
\newcommand{\ETH}{\textnormal{ETH}}

\DeclareMathOperator{\tw}{tw}
\DeclareMathOperator{\bw}{bw}
\DeclareMathOperator{\cw}{cw}

\DeclareMathOperator{\poly}{poly}

\DeclareMathOperator{\col}{col}

\newcommand{\YES}{YES}
\newcommand{\NO}{NO}

\theoremstyle{definition}
\newtheorem{defi}{Definition}

\theoremstyle{plain}
\newtheorem{lem}{Lemma}
\newtheorem{thm}{Theorem}

\newtheorem{corollary}{Corollary}
 \newtheorem{proposition}{Proposition}
\theoremstyle{remark}
\newtheorem*{remark}{Remark}

\newcommand{\calD}{\mathcal{D}}
\newcommand{\calA}{\mathcal{A}}
\newcommand{\calB}{\mathcal{B}}
\newcommand{\N}{\mathbb{N}}

\newcommand{\decprob}[3]{
  \begin{center}
    \begin{minipage}{0.95\textwidth}
      \noindent
      \prob{#1}\vspace{2pt}\\
      \setlength{\tabcolsep}{3pt}
      \begin{tabularx}{\textwidth}{@{}lX@{}}
	\normalsize \textbf{Input:} 		& \normalsize #2 \\
	\normalsize \textbf{Question:} 	& \normalsize #3
      \end{tabularx}
    \end{minipage}
  \end{center}
}

\usepackage{etoolbox}

\newcommand{\keywords}[1]{\noindent \textbf{Keywords}: #1}

\newcommand{\appendixsection}[1]{%
}

\newcommand{\toappendix}[1]{%
  {{
    #1
  }}
}

\newcommand{\noproof}{}

 \title{Diminishable Parameterized Problems and \\ Strict Polynomial Kernelization}%

 \author[1]{Henning~Fernau}
 \author[2]{Till~Fluschnik\thanks{Supported by the DFG, project DAMM (NI 369/13).}$^,$}
 \author[3]{Danny~Hermelin\thanks{Supported by the People Programme (Marie Curie Actions) of the European Union's Seventh Framework Programme (FP7/2007-2013) under REA grant agreement number 631163.11, and by the ISRAEL SCIENCE FOUNDATION (grant No. 551145/14). Also supported by a DFG Mercator fellowship, project DAMM (NI 369/13) while staying at TU Berlin (August 2016).}$^,$}
 \author[4]{Andreas~Krebs\thanks{Supported by the DFG Emmy Noether program (KR 4042/2).}$^,$}
 \author[2]{Hendrik~Molter\thanks{Partially supported by the DFG, project DAPA (NI~369/12).}$^,$}
 \author[2]{Rolf~Niedermeier}
 \affil[1]{Fachbereich 4 -- Abteilung Informatik, Universit\"at Trier, Germany,
 \texttt{fernau@uni-trier.de}}
 \affil[2]{Institut f\"ur Softwaretechnik und Theoretische Informatik,
  TU Berlin, Germany,
  \texttt{\{till.fluschnik,h.molter,rolf.niedermeier\}@tu-berlin.de}}
 \affil[3]{Ben Gurion University of the Negev, Beersheba, Israel,
 \texttt{hermelin@bgu.ac.il}}
 \affil[4]{Wilhelm-Schickard-Institut f\"ur Informatik, Universit\"at T\"ubingen, Germany,
 \texttt{krebs@informatik.uni-tuebingen.de}}

 \date{}

\begin{document}

\maketitle

\looseness=-1
\begin{abstract}

Kernelization---a mathematical key concept for provably effective polynomial-time preprocessing of \NP{}-hard problems---plays a central role in parameterized complexity and has triggered an extensive line of research. This is in part due to a lower bounds framework that allows to exclude polynomial-size kernels under the assumption of \NPnotincoNPslashpoly{}. In this paper we consider a restricted %
yet %
natural variant of kernelization, namely \emph{strict kernelization}, where one is not allowed to increase the parameter of the reduced instance (the kernel) by more than an additive constant. 

Building on earlier work of Chen, Flum, and M\"{u}ller~[Theory~Comput.~Syst.~2011] and developing a general and remarkably %
simple framework, we show that a %
variety of \FPT{} problems does not admit strict polynomial kernels under the %
weaker assumption of \PnotNP{}. 
In particular, we show that various (multicolored) graph problems and Turing 
machine computation problems do not admit strict polynomial kernels unless \PeqNP{}.
To this end, a key concept we use are \emph{diminishable problems}; these are 
parameterized problems that allow to decrease the parameter of the 
input instance by at least one in polynomial time, thereby outputting an equivalent problem instance.
Finally, we study a relaxation of the notion of strict kernels and reveal its limitations.

\smallskip
\end{abstract}

\section{Introduction}
\label{sec:intro}
\looseness=-1

Kernelization is one of the most fundamental concepts in the field of parameterized complexity analysis. 
Given a (w.l.o.g.\ binarily encoded) instance $(x,k) \in \{0,1\}^* \times \mathbb{N}$ of some parameterized problem $L$, a \emph{kernelization} for~$L$ produces in polynomial time an instance~$(x',k')$ satisfying: $(x',k') \in L \iff (x,k) \in L$ and $|x'| + k' \leq f(k)$ for some fixed computable function~$f(k)$. 
In this way, kernelization can be thought of as a preprocessing procedure that reduces an instance to its ``computationally hard core" (\emph{i.e.},~the \emph{kernel}). 
The function~$f(k)$ is accordingly called the \emph{size} of the kernel, and it is typically the measure that one wishes to minimize. Kernelization is a central concept in parameterized complexity not only as an important algorithmic tool, but also because it provides an alternative definition of \emph{fixed-parameter tractability} (\FPT{}): A parameterized problem is solvable in $f(k) \cdot |x|^{O(1)}$ time if and only if it has a kernel of size $g(k)$ for some arbitrary computable functions $f$ and $g$ only depending on the parameter~$k$~\cite{CaiCDF97}.
An algorithm with running time $f(k) \cdot |x|^{O(1)}$ for a parameterized problem $L$ implies that~$L$ has a kernel of size $f(k)$, but in the converse direction one cannot always take the same function~$f$. 
\iflong{}%
  For example, the \NP{}-complete graph problem \prob{Vertex Cover} parameterized by the solution size~$k$ has a $2k$-vertex kernel~\cite{ChenKJ01}, but an algorithm running in $2k \cdot |x|^{O(1)}$ time for the problem obviously would imply \PeqNP{}. 
\fi{}%
The goal of minimizing the size of the problem kernel leads to the question of what is the kernel with the smallest size one can obtain in polynomial time for a given problem. 
In particular, do all fixed-parameter tractable problems have small kernels, say, of linear or polynomial size?

\iflong{}
  The latter question was answered negatively by Bodlaender \emph{et al.}~\cite{BodlaenderDFH09} who used a lemma of Fortnow and Santhanam~\cite{FortnowS11} to show that various \FPT{} problems, \emph{e.g.}~\prob{Path} parameterized by the solution size, do not admit a polynomial-size kernel (or \emph{polynomial kernel} for short) unless \NPincoNPslashpoly{} (which implies a collapse in the Polynomial Hierarchy to its third level). 
  This led to the exclusion of polynomial kernels for various other problems, and the framework of Bodlaender \emph{et al.}~has been extended in several directions~\cite{Kra14}. Regardless, all of these extensions rely on the assumption that \NPnotincoNPslashpoly{}, an assumption that while widely believed in the community, is a much stronger assumption than \PnotNP{}.
\fi{}

Throughout the years, researchers have considered different variants of kernelization such as \emph{Turing kernelization}~\cite{BFFLSV12,SKMN12}, \emph{partial kernelization}~\cite{BGKN11}, \emph{lossy kernelization}~\cite{LokshtanovPRS17}, and \emph{fidelity-preserving preprocessing}~\cite{FKRS12}. 
In this paper, we consider a variant which has been considered previously quite a bit\footnote{We note that in the literature (e.g.~\cite{LinFCL17,XiaoK17}), this definition can also be found for kernelization.}, which is called proper kernelization~\cite{AbuKhzamF06} or \emph{strict kernelization}~\cite{chen2011lower}:
\begin{defi}[Strict Kernel]
\label{def:strictPK}
A strict kernelization for a parameterized problem $L$ is a polynomial-time algorithm that on input instance~$(x,k) \in \{0,1\}^* \times \mathbb{N}$ outputs an instance~$(x',k')\in \{0,1\}^* \times \mathbb{N}$, the \emph{strict kernel}, satisfying:
\begin{inparaenum}[(i)]
\item $(x, k) \in L \iff (x', k') \in L$,
\item $|x'| \le f(k)$, for some function $f$, and
\item $k' \le k + c$, for some constant $c$.
\end{inparaenum}
We say that $L$ admits a strict \emph{polynomial} kernelization if~$f(k)\in k^{O(1)}$.
\end{defi}
\noindent Thus, a strict kernelization is a kernelization that does not increase the output parameter~$k'$ by more than an additive constant.
While the term ``strict'' in the definition above makes sense mathematically, it is actually quite harsh from a practical perspective. Indeed, most of the early work on kernelization involved applying so-called \emph{data reduction rules} that rarely ever increase the parameter value (see~\emph{e.g.}~the surveys~\cite{GN07,Kra14}). %
Furthermore, strict kernelization is clearly preferable to kernelizations that increase the parameter value in a dramatic way: 
Often a fixed-parameter algorithm on the resulting problem kernel is applied, which running time highly depends on the value of the parameter, and so a kernelization that substantially increases the parameter value might in fact be useless. 
Finally, the equivalence with \FPT{} is preserved: A parameterized problem is solvable in $f(k) \cdot |x|^{O(1)}$ time if and only if it has a strict kernel of size $g(k)$.

Chen, Flum, and M\"{u}ller~\cite{chen2011lower} showed that \prob{Rooted Path}, the problem of finding a path of length~$k$ in a graph that starts from a prespecified root vertex, has no strict polynomial kernel unless \PnotNP{}.\footnote{Chen \emph{et al.}~\cite{chen2011lower} did not allow any increase in the parameter. But this is not a crucial difference.} 
They also showed a similar result for \prob{CNF-Sat} parameterized by the number of variables. 
Both of these results seemingly are the only known polynomial kernel lower bounds that rely on the assumption of \PnotNP{} (see Chen \emph{et al.}~\cite{CFKX07} for a few linear lower bounds that also rely on \PnotNP{}). 
The goal of this paper is to show that Chen~\emph{et~al.}'s framework applies for more problems, indeed allowing for a surprisingly simple, natural, and elegant proof framework. 
Indeed, we consider the relative simplicity compared with the standard framework for kernel lower bounds as a particular virtue of our contribution.
Herein, the (mathematical) simplicity manifests itself in:
\begin{compactitem}
 \item The notion of parameter diminisher is easy to grasp.
 \item The correctness of the framework is easy to understand and to prove.
 \item The framework is easy to apply and to extend.
 \item The complexity-theoretic assumption~\PnotNP{} is the gold-standard and employed in many algorithmic contexts.
\end{compactitem}

\paragraph{Our Results.}
We build on the work of Chen \emph{et al.}~\cite{chen2011lower}, and further develop and widen the framework they presented for excluding strict polynomial kernels. Using this extended framework, we show that several natural parameterized problems in \FPT{} have no strict polynomial kernels under the assumption that \PnotNP{}. 
In particular, the main result of this paper is given in Theorem~\ref{thm:main} below. 
Note that we use the brackets in the problem names to denote the parameter under consideration.\footnote{For a complete list of problem definitions see \Cref{sec:probzoo}.} 
Thus, \pprob{Multicolored Path}{k} is the \prob{Multicolored Path} problem parameterized by the solution size $k$, for instance. %
These conventions will be used throughout our paper. 
\begin{thm}
\label{thm:main}%
Unless \PeqNP{}, each of the following fixed-parameter tractable problems does not admit a strict polynomial kernel:
\begin{compactitem}
\item \pprob{Multicolored Path}{k} and \pprob{Multicolored Path}{k \log n};
\item \pprob{Clique}{\Delta}\noproof{}, \pprob{Clique}{\tw}\noproof{}, \pprob{Clique}{\bw}\noproof{}, and \pprob{Clique}{\cw};
\item \pprob{Biclique}{\Delta}\noproof{}, \pprob{Biclique}{\tw}\noproof{}, \pprob{Biclique}{\bw}\noproof{}, and \pprob{Biclique}{\cw}\noproof{};
\item \pprob{Colorful Graph Motif}{k} and \pprob{Terminal Steiner Tree}{k+|T|}\noproof{};%
\item \pprob{Short NTM Computation}{k+|\Sigma|} and \pprob{Short NTM Computation}{k+|Q|}.
\item \pprob{Short Binary NTM Computation}{k}\noproof{};
\end{compactitem}
Herein, $k$~denotes the solution size, $n$~denotes the number of vertices 
in the graph, %
$\Delta$ denotes the maximum vertex degree in the graph, 
$\tw$, $\bw$, and $\cw$~denote the treewidth, bandwidth, and cutwidth of the graph, respectively, 
$T$~denotes the set of terminals,
$|\Sigma |$~denotes the alphabet size, and~$|Q|$~denotes the number of states.
\end{thm}

We give formal definitions of each of these parameterized problems in the following sections. 
For now, let us mention that several of these have prominent roles in previous kernelization lower bound papers.
  For instance, \pprob{Multicolored Path}{k} is a WK[1]-complete problem~\cite{HermelinKSWW15}. 
  The \prob{Colorful Graph Motif} problem has been used to show that several problems in degenerate graphs have no polynomial kernels unless \NPnotincoNPslashpoly{}~\cite{CyganPPW12}
  .
Finally, we also explore how ``tight'' the concept of strict polynomial 
kernels is and, employing the Exponential Time Hypothesis (ETH), 
conclude that we often cannot hope for significantly relaxing the 
concept of strict kernelization to achieve a comparable list of such analogues kernel lower bounds 
under \PnotNP{}.

The main concept behind the new proof framework is that of a \emph{parameter diminisher}: an algorithm that is able to decrease the parameter value of any given instance by at least one in polynomial time. This concept was first observed by Chen, Flum, and M\"{u}ller~\cite{chen2011lower} who called it a \emph{parameter-decreasing polynomial self-reduction}. It is not difficult to show that the existence of a parameter diminisher and a strict polynomial kernel for an \NP{}-hard parameterized problem implies \PeqNP{}. But surprisingly enough, as we will show, there are numerous natural diminishable problems. 
And excluding strict polynomial kernelization is comparatively simple for these problems. 
We remark that for the problems we discuss in this paper, one can exclude polynomial kernels under the assumption that \NPnotincoNPslashpoly{} using the \iflong{}framework of Bodlaender \emph{et al.}~\cite{BodlaenderDFH09}, which is based on a rather indirect lemma of Fortnow and Santhanam~\cite{FortnowS11}\fi{}. 
\iflong{}%
  Our results base on a weaker assumption, but exclude a more restricted version of polynomial kernels. 
  On the contrary, Bodlaender \emph{et al.}'s framework excludes a more general version of polynomial kernels but requires a stronger assumption. 
\fi{}%
Hence, our results are incomparable with the existing no-polynomial-kernel results. 
However, the new framework provides a simpler methodology and directly connects the exclusion of strict polynomial kernels to the assumption that \PnotNP{}.

\iflong{}
\paragraph{Outline.}
The paper is organized as follows. In Section~\ref{sec:framework} we present the basic framework of Chen, Flum, and M\"{u}ller~\cite{chen2011lower}, and further develop and widen it so that we have all necessary tools for proving Theorem~\ref{thm:main}. In particular, we formally define the concept of parameter diminisher and diminishable problems. In Section~\ref{sec:nopk}, we provide our
list of problems without strict polynomial kernels, essentially
proving Theorem~\ref{thm:main}. In Section~\ref{sec:nodiminishers}
we investigate the effect of going from strict polynomial kernels
(parameter may only increase by an additive constant) to ``semi-strict''
polynomial kernels (parameter may increase by a constant factor). We see
that only few problems so far allow for this while we also show that often  
the corresponding ``strong diminishers'' (which would yield 
semi-strict polynomial kernels) do not exist unless the ETH breaks. 
We conclude in Section~\ref{sec:conclusion}.%
\fi{}

\paragraph{Notation.}
We use basic notation from parameterized complexity~\cite{cygan2015parameterized,DowneyF13,FlumG06,Niedermeier06} and graph theory~\cite{Diestel10,west}.
Let $G=(V,E)$ be a graph. For a vertex set~$W\subseteq V$, let $G-W:=(V\setminus W,\{e\in E\mid e\cap W=\emptyset\})$. 
If $W=\{v\}$, then we write $G-v$ instead of $G-\{v\}$. Furthermore, we denote the induced subgraph on vertices $W$ by $G[W]$.
For a vertex~$v\in V$, we denote by~$N_G(v):=\{w\in V\mid \{v,w\}\in E\}$ the neighborhood of~$v$ in~$G$.
If not specified differently, then we denote by~$\log$ the logarithm with base two.

\section{Framework}
\label{sec:framework}
\looseness=-1
\appendixsection{sec:framework}

In this section we present the general framework that we will use throughout the paper. We formally define the notion of a parameter diminisher which is central to the entire paper and we show why this concept strongly relates to the exclusion of polynomial kernels.

\begin{defi}[Parameter Diminisher]
\label{def:diminisher}
A \emph{parameter diminisher} for a parameterized problem~$L$ is a polynomial-time algorithm that maps instances~$(x,k) \in \{0,1\}^* \times \mathbb{N}$ to instances~$(x', k') \in \{0,1\}^* \times \mathbb{N}$ such that $(x, k) \in L \text{ if and only if } (x', k') \in L$ and $k' < k$.
\end{defi}
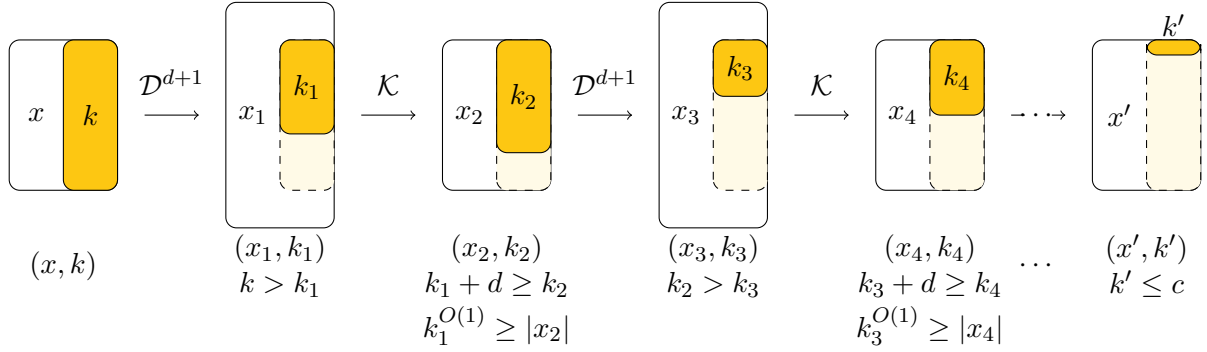
\begin{figure}
 \centering
  \begin{tikzpicture}

  \def\xsc{0.72}
  \def\x{4.0*\xsc}
	  \draw[rounded corners, fill=white] (0, 0) rectangle (2*\xsc, 2);
      \node (i) at (0.5* \xsc,1)[]{$x$};
      \draw[rounded corners,fill=lipyel] (1*\xsc, 0) rectangle (2*\xsc, 2);
      \node (i) at (1.5*\xsc,1)[]{$k$};
	  \node (i) at (1*\xsc,-1)[]{$(x,k)$};

  \draw[->] (2.5*\xsc,1) to node[label=90:{$\calD^{d+1}$}]{}(3.5*\xsc,1);

  \def\mul{1}
		  \draw[rounded corners, fill=white] (0+\x, 0-0.5) rectangle (\x+2*\xsc, 2.5);
      \node (i) at ( \x+0.5* \xsc,1)[]{$x_1$};
	\draw[rounded corners,dashed,fill=lipyel!10!white] (\mul*\x+1*\xsc, 0) rectangle (\mul*\x+2*\xsc, 2);
      \draw[rounded corners,fill=lipyel] (\x+1*\xsc, 0.75) rectangle (\x+2*\xsc, 2);
      \node (i) at (\x+1.5*\xsc,1.375)[]{$k_1$};
	  \node (i) at (\x+1*\xsc,-1)[align=center]{$(x_1,k_1)$ \\ $k>k_1$};

  \draw[->] (\mul*\x+2.5*\xsc,1) to node[label=90:{$\mathcal{K}$}]{}(\mul*\x+3.5*\xsc,1);

  \def\mul{2}
		  \draw[rounded corners, fill=white] (\mul*\x+0, 0) rectangle (\mul*\x+2*\xsc, 2);
      \node (i) at (\mul*\x+0.5* \xsc,1)[]{$x_2$};
	\draw[rounded corners,dashed,fill=lipyel!10!white] (\mul*\x+1*\xsc, 0) rectangle (\mul*\x+2*\xsc, 2);
      \draw[rounded corners,fill=lipyel] (\mul*\x+1*\xsc, 0.5) rectangle (\mul*\x+2*\xsc, 2);
      \node (i) at (\mul*\x+1.5*\xsc,1.25)[]{$k_2$};
	  \node (i) at (\mul*\x+1*\xsc,-1)[align=center]{$(x_2,k_2)$ \\ $k_1+d\geq k_2$};
	   \node (i) at (\mul*\x+1*\xsc,-1.8)[align=center]{$k_1^{O(1)}\geq |x_2|$};
  
  \draw[->] (\mul*\x+2.5*\xsc,1) to node[label=90:{$\calD^{d+1}$}]{}(\mul*\x+3.5*\xsc,1);
  \def\mul{3}
		  \draw[rounded corners, fill=white] (\mul*\x+0, 0-0.5) rectangle (\mul*\x+2*\xsc, 2.5);
      \node (i) at (\mul*\x+0.5* \xsc,1)[]{$x_3$};
	\draw[rounded corners,dashed,fill=lipyel!10!white] (\mul*\x+1*\xsc, 0) rectangle (\mul*\x+2*\xsc, 2);
	  \draw[rounded corners,fill=lipyel] (\mul*\x+1*\xsc, 1.25) rectangle (\mul*\x+2*\xsc, 2);

      \node (i) at (\mul*\x+1.5*\xsc,1.625)[]{$k_3$};
	  \node (i) at (\mul*\x+1*\xsc,-1)[align=center]{$(x_3,k_3)$ \\ $k_2> k_3$};

  \draw[->] (\mul*\x+2.5*\xsc,1) to node[label=90:{$\mathcal{K}$}]{}(\mul*\x+3.5*\xsc,1);
  \def\mul{4}
		  \draw[rounded corners, fill=white] (\mul*\x+0, 0) rectangle (\mul*\x+2*\xsc, 2);
      \node (i) at (\mul*\x+0.5* \xsc,1)[]{$x_4$};
	\draw[rounded corners,dashed,fill=lipyel!10!white] (\mul*\x+1*\xsc, 0) rectangle (\mul*\x+2*\xsc, 2);
	  \draw[rounded corners,fill=lipyel] (\mul*\x+1*\xsc, 1) rectangle (\mul*\x+2*\xsc, 2);

      \node (i) at (\mul*\x+1.5*\xsc,1.5)[]{$k_4$};
	  \node (i) at (\mul*\x+1*\xsc,-1)[align=center]{$(x_4,k_4)$\\ $k_3+d\geq k_4$};
	  \node (i) at (\mul*\x+1*\xsc,-1.8)[align=center]{$k_3^{O(1)}\geq |x_4|$};

  \draw[-] (\mul*\x+2.5*\xsc,1) to (\mul*\x+2.75*\xsc,1);
  \node (ld) at (\mul*\x+3*\xsc,1)[]{$\cdots$};
  \draw[->] (\mul*\x+3.25*\xsc,1) to (\mul*\x+3.5*\xsc,1);
  \node (ld) at (\mul*\x+3*\xsc,-1)[]{$\cdots$};

  \def\mul{5}
		  \draw[rounded corners, fill=white] (\mul*\x+0, 0) rectangle (\mul*\x+2*\xsc, 2);
      \node (i) at (\mul*\x+0.5* \xsc,1)[]{$x'$};
	\draw[rounded corners,dashed,fill=lipyel!10!white] (\mul*\x+1*\xsc, 0) rectangle (\mul*\x+2*\xsc, 2);
	  \draw[rounded corners,fill=lipyel] (\mul*\x+1*\xsc, 1.8) rectangle (\mul*\x+2*\xsc, 2);

      \node (i) at (\mul*\x+1.5*\xsc,2.2)[]{$k'$};
	  \node (i) at (\mul*\x+1*\xsc,-1)[align=center]{$(x',k')$ \\ $k'\leq c$};

  \end{tikzpicture}
  \caption{Illustration to the proof of \Cref{thm:diminisher} for an input instance~$(x,k)$. Herein, $\mathcal{K}$ and $\mathcal{D}$ denote the strict kernelization with additive constant~$d$ and the parameter diminisher, respectively. We represent each instance by boxes, where the size of a box indicates the size of the instance or the value of the parameter (each dashed box refers to~$k$).}
  \label{fig:dim}
\end{figure}
Thus, a parameter diminisher is an algorithm that is able to decrease the parameter of any given instance of a parameterized problem~$L$ in polynomial time. 
\iflong{}The algorithm is given freedom in that it can produce a completely different instance, as long as its an equivalent one (with respect to~$L$) and has a smaller parameter value.\fi{}%
We call a parameterized problem~$L$~\emph{diminishable} if there is a parameter diminisher for $L$. 
The following theorem was proved initially by Chen, Flum, and M\"{u}ller~\cite{chen2011lower}, albeit for slightly weaker forms of diminisher and strict polynomial kernels.

\begin{thm}[\cite{chen2011lower}]
\label{thm:diminisher}%
Let $L$ be a parameterized problem such that its unparameterized version is \NP{}-hard and we have that $\{(x,k)\in L\mid k\le c\}\in \PTIME{}$, for some constant $c$. %
If $L$ is diminishable and admits a strict polynomial kernel, then \PeqNP{}.
\end{thm}

\noindent The idea behind~\Cref{thm:diminisher} is to repeat the following two procedures until the parameter value drops below~$c$ (see \Cref{fig:dim} for an illustration). 
First, apply the parameter diminisher a constant number of times such that when, second, the strict polynomial kernelization is applied, the parameter value is decreased. The strict polynomial kernelization keeps the instances small, hence the whole process runs in polynomial time.
Reductions transfer diminishability from one parameterized problem to another if they do not increase the parameter value and run in polynomial time.

\begin{defi}[Parameter-Non-Increasing Reduction]
 Given two parameterized problems~$L$ with parameter~$k$ and~$L'$ with parameter~$k'$, a \emph{parameter-non-increasing reduction} from~$L$ to~$L'$ is an algorithm that maps each instance $(x,k)$ of $L$ to an equivalent 
 instance~$(x',k')$ of~$L'$ in time polynomial in $|x|+k$ such that $k'\leq k$.
\end{defi}
Note that, in order to transfer diminishability, we need parameter-non-increasing reductions between two parameterized problems in both directions. This is a crucial difference to other hardness results based on reductions.

\begin{lem}%
  \label{lem:dimreduction}
Let~$L_1$ and~$L_2$ be two parameterized problems such that there are parameter-non-increasing reductions from~$L_1$ to~$L_2$ and from~$L_2$ to~$L_1$.
Then we have that~$L_1$ is diminishable if and only if~$L_2$ is diminishable. 
\end{lem}

{
  \begin{proof}
    Let~$L_1$ with parameter~$k_1$ and~$L_2$ with parameter~$k_2$ be two parameterized problems.
    Let~$\calA_1$ and $\calA_2$ be parameter-non-increasing reductions from~$L_1$ to~$L_2$ and from~$L_2$ to~$L_1$, respectively.
    Let~$\calD_2$ be a parameter diminisher for~$L_2$. 
    Let $(x_1,k_1)$ be an arbitrary instance of~$L_1$.
    
    Apply $\calA_1$ to $(x_1,k_1)$ to obtain the instance $(x_2,k_2)$ of~$L_2$ with $k_2\leq k_1$.
    Next, apply $\calD_2$ to~$(x_2,k_2)$ to obtain the instance $(x_2',k_2')$ of~$L_2$ with $k_2'<k_2$.
    Finally, apply $\calA_2$ to $(x_2',k_2')$ to obtain the instance~$(x_1',k_1')$ of~$L_1$ with $k_1'\leq k_2'$.
    As $k_1'\leq k_2'<k_2 \leq k_1$, the above combination of $\calA_1$, $\calD_2$, and~$\calA_2$ forms a parameter diminisher for~$L_1$.
    To get the reverse direction, exchange the roles of~$L_1$ and~$L_2$.
  \end{proof}
}

\paragraph{Parameter-Decreasing Branching and Strict Composition.}
To construct parameter diminisher, it is useful to follow a ``branch and compose'' technique: Herein, first \emph{branch} into several subinstances while decreasing the parameter value in each, and then \emph{compose} the subinstances into one instance without increasing the parameter value by more than a constant. 
We first give the definition of branching rule and composition and then prove that those two algorithms combined form a parameter diminisher.

Branching rules are highly common in parameterized algorithm design, and they are typically deployed when using depth-bounded search-tree or related techniques. Roughly speaking, in a \emph{parameter-decreasing branching rule} one reduces the problem instance to several problem instances with smaller parameters such that at least one of these new instances is a yes-instance if and only if the original instance is a yes-instance. 
\begin{defi}[Parameter-Decreasing Branching Rule]
\label{def:branching}
A \emph{parameter-decreasing branching rule} for a parameterized problem $L$ is a polynomial-time algorithm that on input $(x,k) \in \{0,1\}^* \times \mathbb{N}$ outputs a sequence of instances $(y_1,k'),\ldots,(y_t,k') \in \{0,1\}^* \times \mathbb{N}$ such that~$(x,k) \in L \iff (y_i,k') \in L$ for some $i \in \{1,\ldots,t\}$ and $k' < k$.
\end{defi}
Composition is the core concept behind the standard kernelization lower bound framework introduced by Bodlaender \emph{et al.}~\cite{BodlaenderDFH09}. Here we use a more restrictive notion of this concept:

\begin{defi}[Strict Composition]
\label{def:compostion}%
A \emph{strict composition} for a parameterized problem $L$ is an algorithm that receives as input $t$ instances $(x_1,k),\ldots,(x_t,k)\in \{0,1\}^* \times \mathbb{N}$, and outputs in polynomial~time a single instance $(y,k')\in \{0,1\}^* \times \mathbb{N}$ such that
\begin{inparaenum}[(i)]
\item $(y,k') \in L \iff (x_i,k) \in L$ for some $i \in \{1,\ldots,t\}$ and
\item $k' \leq k + c$ for some constant $c$.
\end{inparaenum}
\end{defi}
If we now combine (multiple applications of) a parameter-decreasing branching rule with a strict composition, then we get a parameter diminisher. 

\begin{lem}%
\label{lem:branchcomp}%
Let $L$ be a parameterized problem. If $L$ admits a parameter-decreasing branching rule and a strict composition, then it is diminishable.
\end{lem}
{
\begin{proof}
Let $(x, k)$ be an instance of a parameterized problem $L$ and $c$ be the constant associated with a strict composition for~$L$. 
We recursively apply the parameter-decreasing branching rule for~$L$ $c+1$ times to produce $t^{c+1}$ instances $(x_i,k^*)$ with $k^* < k-c$, since every time the parameter is decreased at least by one. 
The strict composition, receiving $t^{c+1}$ instances, produces an instance~$(y, k')$ with~$k'\leq k^*+c<k$ of $L$ which is a yes-instance if and only if $(x, k)$ is a yes-instance. 
Hence, the whole procedure is a parameter diminisher for~$L$.
\end{proof}
}
We remark that \Cref{lem:branchcomp} also holds if we require in both \Cref{def:compostion}(i) and \Cref{def:branching} that the equivalence hold for all $i \in \{1,\ldots,t\}$.
As an example application of \Cref{lem:branchcomp} above, we consider the parameter diminisher for the \pprob{Rooted Path}{k} problem due to Chen~\emph{et al.}~\cite{chen2011lower}. 
In this problem we are given an undirected graph~$G=(V,E)$, a distinguished vertex~$r \in V$, and an integer $k$, and our goal is to determine whether there exists a simple path in $G$ of length $k$ that starts at~$r$. 
Let $v_1,\ldots,v_t$ be the neighbors of~$r$ in $G$. 
The parameter-decreasing branching rule for \pprob{Rooted Path}{k} constructs from $(G,r,k)$ the set of instances~$(G-r,v_1,k-1),\ldots,(G-r,v_t,k-1)$. 
A strict composition for \pprob{Rooted Path}{k} takes as input the instances $(G_1,r_1,k),\ldots,(G_t,r_t,k)$ and constructs the instance~$(G',r,k+1)$, where~$G'$ is the graph obtained by taking the disjoint union of all $G_i$s and making all their roots adjacent to a new root vertex $r$. 
Combining these two algorithms gives the parameter diminisher for \pprob{Rooted Path}{k}.%

\paragraph{On the Exclusion of Non-Uniform Kernelization Algorithms.}
We want to point out that the framework can even be used to exclude strict polynomial kernels computed in \emph{non-uniform} polynomial time (the corresponding complexity class is called $\PTIME / \text{poly}$). With the framework, we can exclude non-uniform strict polynomial kernels under the assumption that $\NP \nsubseteq \PTIME/\text{poly}$. 
\begin{defi}[Non-Uniform Strict Kernel]
\label{def:strictPK}
A non-uniform strict kernelization for a parameterized problem $L$ is a \emph{non-uniform} polynomial-time algorithm that on input instance~$(x,k) \in \{0,1\}^* \times \mathbb{N}$ outputs an instance~$(x',k')\in \{0,1\}^* \times \mathbb{N}$, the \emph{strict kernel}, satisfying:
\begin{inparaenum}[(i)]
\item $(x, k) \in L \iff (x', k') \in L$,
\item $|x'| \le f(k)$, for some function $f$, and
\item $k' \le k + c$, for some constant $c$.
\end{inparaenum}
We say that $L$ admits a non-uniform strict \emph{polynomial} kernelization if $f(k)\in k^{O(1)}$.
\end{defi}
\begin{proposition}
\label{prop:non-uniform}
Let $L$ be a parameterized problem such that its unparameterized version is \NP{}-hard and we have that $\{(x,k)\in L\mid k\le c\}\in \PTIME{}/\textnormal{poly}$, for some constant $c$. %
If $L$ is diminishable and admits a non-uniform strict polynomial kernel, then $\NP \subseteq \PTIME/\textnormal{poly}$.
\end{proposition}
We remark that if $\NP \subseteq \PTIME/\textnormal{poly}$, then the Polynomial Hierarchy collapses to its second level~\cite{karp1982turing} (recall that \NPincoNPslashpoly{} implies a collapse in the Polynomial Hierarchy to its third level).
\begin{proof}
 Let $L$ be a parameterized problem whose unparameterized version is \NP{}-hard and it holds that $\{(x,k)\in L\mid k\le c\}\in \PTIME{}/\textnormal{poly}$, for a constant $c\geq 1$. 
 Let $\calD$ be a parameter diminisher for $L$ with constant $c_d\geq 1$ and let $\calA$ be a non-uniform strict polynomial kernelization for~$L$ with constant~$c_a>1$. 
 \iflong{}We show that we can solve any instance $(x,k)$ of $L$, with $k$ being the parameter, in non-uniform polynomial time. \fi{}%
 Let $(x,k)$ be an instance of $L$.
 Apply $\calD$ on $(x,k)$ exactly $c_r:=\lceil (c_a+1)/c_d\rceil$~times to obtain an equivalent instance $\calD^{c_r}(x,k)=(x',k')$ with $k'\leq k-c_d \cdot c_r< k-c_a$.
 Observe that the size of~$|(x', k)|$ is still polynomial in $|(x, k)|$ as $c_r$ is a constant.
 Next, apply $\calA$ on $(x',k')$ to obtain an equivalent instance $(x'',k'')$ with $|(x'', k'')|\leq k'^{c'}$, $c'\geq 1$, and $k''\leq k' + c_a <k$; we explain how to get the advice for $\calA$ in the second half of the proof.
 Repeating the described procedure at most~$k$ times produces an instance $(y,1)$ of $L$, solvable in non-uniform polynomial time. 
 
In the remainder of the proof, we explain how we obtain the advices for the applications of the non-uniform strict polynomial kernelization. 
Given any instance size $n$, we claim that for all instances $(x, k) \in L$ with $|(x, k)| = n$, we can upper-bound the sizes of all instances $(x^*, k^*)$ for which a non-uniform strict polynomial kernel needs to be computed in the above described procedure by a polynomial in $n$. 
Let the running time of the diminisher~$\calD$ be $O(n^{c''})$.
It follows that~$|(x',k')| \le c_O\cdot n^{c''\cdot c_r}$, where~$(x',k')=\calD^{c_r}(x,k)$ and~$c_O$ is the constant hidden in the $O$-notation. 
Furthermore, we have that the size of all subsequent instances the diminisher sequence $\calD^{c_r}$ is applied to is smaller than~$k^{c'}$. 
Thus, the sizes of all instances $(x^*, k^*)$ for which a non-uniform strict polynomial kernel needs to be computed can be upper-bounded by $c_O\cdot n^{c' \cdot c''\cdot c_r}$. 
A list containing all advices for all instance sizes smaller or equal to this bound has polynomial length, since all advices have polynomial length. 
Hence, the overall procedure can take that list as an advice and use it as a look-up table for the advices needed to apply the non-uniform strict polynomial kernels.
\end{proof}

\section{Problems without Strict Polynomial Kernels}
\label{sec:nopk}
\appendixsection{sec:nopk}
\looseness=-1

In this section we prove Theorem~\ref{thm:main} based on several propositions to follow. We present parameter diminishers for most problems mentioned in Theorem~\ref{thm:main}, and for some problems we show that they do not admit strict polynomial kernels (unless \PeqNP{}) using parameter-non-increasing reductions. Note that all parameterized problems we consider are known to be in \FPT{}. 
This section is segmented into three parts, the first dealing with
\prob{Clique}, \textsc{Biclique}, and \prob{Terminal Steiner Tree}, the second with multicolored graph problems,
and the third part deals with non-deterministic Turing machine
computations.
\paragraph{Clique, Biclique, and a Steiner Tree Problem.}\label{ssec:cliqedom}

We begin with the \prob{Clique} problem: Given an undirected graph $G=(V,E)$ and an integer~$k$, determine whether there exist $k$ vertices~$v_1,\ldots,v_k \in V$ such that $\{v_i,v_j\} \in E$ for all~$1\le i < j \le k$. 
Since \pprob{Clique}{k} is \W{1}-complete, we focus on other parameterizations of \prob{Clique} that are contained in~\FPT{}, for instance the maximum degree~$\Delta$ of the input graph, where \prob{Clique} has a simple \FPT{} algorithm\iflong{}: Exhaustively search the closed neighborhood of each vertex individually\fi{}~\cite{cygan2015parameterized}. 
Other parameterizations include treewidth~$\tw=\tw(G)$, bandwidth $\bw=\bw(G)$, and the cutwidth $\cw=\cw(G)$ of the input graph. 
As cutwidth is at least as large as all these parameters in every graph, for our purposes we only recall the definition of cutwidth: 
A graph~$G=(V,E)$ has \emph{cutwidth} at most~$k$ if and only if there exists a linear ordering (layout) $\pi: V \to \{1,\ldots,|V|\}$ of~$G$ where for each real number $1<\alpha<|V|$ we have that for the \emph{cut} at $\alpha$, defined as~$E_{\pi,\alpha}:=\{ \{u,v\} \in E: \pi(u) < \alpha < \pi(v)\}$, it holds that~$|E_{\pi,\alpha}| \leq k$.

\begin{proposition}
\label{lem:cliquecw}
\pprob{Clique}{\cw} is diminishable.
\end{proposition}
\begin{proof}
Let $(G=(V, E), \cw(G), k)$ be an instance of \pprob{Clique}{\cw}. 
The following is a parameter-decreasing branching rule for $(G=(V, E),\cw(G),k)$: 
For each $v \in V$, let $G_v$ denote the subgraph of~$G$ induced by the open neighborhood of $v$, that is, $G_v=G[N(v)]$.
It is not difficult to see that~$G$ has a clique of size $k$ if and only if some $G_v$ has a clique of size $k-1$. We next argue that $\cw(G_v) < \cw(G)$ for all $v \in V$. 
Consider an optimal linear ordering of~$G$, and for any vertex~$v \in V$, consider the ordering restricted to the closed neighborhood of~$v$ in the optimal ordering. 
By removing~$v$, we remove at least one edge in any cut of the linear ordering, giving us a linear ordering of $G_v$ which has cutwidth smaller than the optimal ordering for $G$. 
Hence, $\cw(G_v) < \cw(G)$ for each~$v \in V$, and the above algorithm is indeed a parameter-decreasing branching rule for \pprob{Clique}{\cw}.

For the composition step we take the disjoint union of all graphs. Since a graph has a clique of size $k$ if and only if one of its connected components has a clique of size $k$, and since the cutwidth of a graph $G$ equals the maximum cutwidth of its connected components, this is a strict composition for \pprob{Clique}{\cw}. 
Applying \Cref{lem:branchcomp}, this completes the proof.
\end{proof}
\iflong{}
Note that the above procedure is also a parameter diminisher for parameters $\Delta$, $\tw$, $\bw$, and $k$. 
Leaving the \W{1}-complete \pprob{Clique}{k} problem aside, we get the following corollary.
\begin{corollary}
  \pprob{Clique}{\Delta}, \pprob{Clique}{\tw}, and \pprob{Clique}{\bw} are diminishable.
  \end{corollary}

  It is easy to show that the parameter diminisher presented for \prob{Clique} can be adapted to the \textsc{Biclique} problem: Given an undirected bipartite graph $G=(A \uplus B, E)$ and an integer $k$, find~$k$ vertices $a_1,\ldots,a_k \in A$ and $b_1,\ldots,b_k\in B$ such that $\{a_i,b_j\} \in E$ for all $1 \leq i , j \leq k$. Thus, we have:

  \begin{corollary}
  \label{lem:bicliquecw}
  \pprob{Biclique}{\Delta}, \pprob{Biclique}{\tw}, \pprob{Biclique}{\bw}, and \pprob{Biclique}{\cw} are diminishable.
  \end{corollary}
\fi{}
We next consider an \NP-complete variant of the well-known \textsc{Steiner Tree} problem: given an undirected graph $G=(V=N\uplus T,E)$ ($T$ is called the terminal set) and a positive integer~$k$, decide whether there is a subgraph~$H\subseteq G$ with at most $k+|T|$ vertices such that~$H$ is a tree containing all vertices in~$T$. 
The variant we consider is the \textsc{Terminal Steiner Tree (TST)}~\iflong{}\cite{LinX02,BiniazMS15}\fi{} problem, which additionally requires the terminal set~$T$ to be a subset of the set of leaves of the tree~$H$. 
For the sake of completeness, we prove the following.

\begin{lem}
\pprob{Terminal Steiner Tree}{k+|T|} is fixed parameter tractable. 
\end{lem}

\begin{proof}
We give an FPT-reduction from \pprob{TST}{k+|T|} to \pprob{Steiner Tree}{k'+|T|}. As \pprob{Steiner Tree}{k'+|T|} is fixed-parameter tractable~\cite{DreyfusW71}, the claim follows.

Let $(G=(N\uplus T,E),k)$ be an instance of \pprob{TST}{k+|T|}.
We construct an equivalent instance~$(G'=(N'\uplus T,E'),k')$ of \pprob{Steiner Tree}{k'+|T|} as follows.
Let $G'$ be initially a copy of~$G$.
For each~$t\in T$, apply the following.
For each edge~$\{v,t\}\in E$, remove~$\{v,t\}$ from~$G'$ and add a path of length~$2(k+|T|)$ to~$G'$ with endpoints~$v$ and~$t$.
Set~$k'=|T|\cdot(2(k+|T|)-1)+k$.
This finishes the reduction.
Clearly, the construction can be done in FPT-time.

We show that $(G=(N\uplus T,E),k)$ is a yes-instance of \pprob{TST}{k+|T|} if and only if $(G'=(N'\uplus T,E'),k')$ is a yes-instance of \pprob{Steiner Tree}{k'+|T|}.

Let~$H$ be a terminal Steiner tree in~$G$ with $\ell\leq k+|T|$ vertices.
We construct a Steiner tree~$H'$ in~$G'$ from~$H$ with $\ell'\leq k'+|T|$ vertices as follows.
Recall that each~$t\in T$ has exactly one neighbor~$v_t$ in~$H$.
Hence, obtain~$H'$ by replacing for each~$t\in T$ the edge~$\{v_t,t\}\in E(H)$ by the path of length~$2(k+|T|)$ connecting~$v_t$ with~$t$.
It is not difficult to see~$H'$ is a Steiner tree in~$G'$.
Moreover, $\ell'=|V(H')|=|V(H)|+|T|(2(k+|T|)-1)\leq k+|T|+|T|(2(k+|T|)-1)=k'+|T|$.

Conversely, let~$H'$ be a minimum Steiner tree in~$G'$ with $\ell'\leq k'+|T|$ vertices.
We state some first observations on~$H'$.
Observe that no inner vertex of the paths added in the construction step from~$G$ to~$G'$ is a leaf of~$H'$ (as otherwise~$H'$ is not minimum).
Moreover, as~$H'$ contains each~$t\in T$,~$H'$ contains a path of length~$2(k+|T|)$ for each~$t\in T$.
Suppose that there is a terminal~$t\in T$ such that~$t$ is not a leaf in~$H'$.
Then~$H'$ contains at least~$|T|+(|T|+1)\cdot(2(k+|T|)-1)=|T|\cdot(2(k+|T|)-1)+2(k+|T|)-1=k'+(k+|T|)-1>k'+|T|$, yielding a contradiction.
Hence, each terminal~$t\in T$ forms a leaf in~$H'$.
We show how to obtain a terminal Steiner tree~$H$ in~$G$ from~$H'$ with at most~$k+|T|$ vertices.
As each terminal~$t\in T$ forms a leaf in~$H'$, there is exactly one neighbor~$v_t\in N_G(t)$ such that~$H'$ contains the path of length~$2(k+|T|)$ connecting~$v_t$ with~$t$.
Replace for each~$t\in T$ the path of length~$2(k+|T|)$ connecting~$v_t$ with~$t$ in~$H'$ by the edge~$\{v_t,t\}\in E$ to obtain~$H$ from~$H'$.
Note that~$H$ is a terminal Steiner tree.
Moreover, $|V(H)|=|V(H')|-|T|(2(k+|T|)-1)\leq |T|+|T|\cdot(2(k+|T|)-1)+k-|T|(2(k+|T|)-1)=k+|T|$.
\end{proof}

We can reduce \pprob{Steiner tree}{k+|T|} to \pprob{Terminal Steiner Tree}{k'+|T|} by adding to each terminal a pendant leaf.
Set the terminal set in TST to the set of added leaf vertices, and ask for a terminal Steiner tree of size $k'=k+|T|$.
This is a polynomial parameter transformation, and hence refutes the existence of a polynomial problem kernel~\cite{DLS14} for TST under \NPnotincoNPslashpoly.

\begin{proposition}%
\label{lem:tstdim}
\pprob{Terminal Steiner Tree}{k+|T|} is diminishable.
\end{proposition}

{
  \begin{proof}
  We present a parameter-decreasing branching rule and a strict composition for \pprob{Terminal Steiner Tree}{k+|T|}. 
  Together with~\Cref{lem:branchcomp}, the claim then follows.
  Let~$(G=(N\uplus T,E), k)$ be an instance of \pprob{Terminal Steiner Tree}{k+|T|} (we can assume that $G$ has a connected component containing~$T$).
  
  We make several assumptions first.
  We can assume that~$|T|\geq 3$ (otherwise a shortest path is the optimal solution) and additionally that for all terminals $t\in T$ it holds that~$N_G(t)\not\subseteq T$ (as otherwise the instance is a no-instance).
  Moreover, we can assume that there is no vertex~$v\in N$ such that $T\subseteq N_G(v)$, as otherwise we immediately output whether $k\geq 1$.
  
  Select a terminal $t^*\in T$, and let $v_1,\ldots,v_d$ denote the neighbors of~$t^*$ in~$G-(T\setminus\{t^*\})$.
  We create $d$ instances $(G_1,k-1),\ldots,(G_d,k-1)$ as follows.
  Define~$G_i$, $i\in[d]$, by $G_i:=G-v_i$. 
  Turn the vertices in~$N_G(v_i)$ in~$G_i$ into a clique, that is, for each distinct vertices $v,w\in N_G(v_i)$ add the edge~$\{v,w\}$ if not yet present. 
  This finishes the construction of~$G_i$.
  It is not hard to see that the construction can be done in polynomial time.

  We show that~$G$ has a terminal Steiner tree of size~$k+|T|$ if and only if there is an~$i\in[d]$ such that~$G_i$ admits a terminal Steiner tree of size~$k-1+|T|$. 
  Suppose that~$G$ has a terminal Steiner tree~$H$ of size~$k+|T|$.
  As~$t^*$ is a leaf in~$H$, there is exactly one neighbor $v_i$, $i\in[d]$, being the neighbor of~$t^*$ in~$H$.
  Let~$w$ be a neighbor of~$v_i$ in~$H-T$ and let $A:=N_H(v_i)$ (note that~$t^*\in A$).
  Then~$H_i$ forms a terminal Steiner tree in~$G_i$, where~$H_i$ is the tree obtained from~$H$ by deleting~$v_i$ and connecting~$w$ with all vertices in~$A$.
  Moreover, $H_i$ is of size~$k-1+|T|$.

  Conversely, let $G_i$ admit a terminal Steiner tree~$H_i$ of size~$k-1+|T|$.
  As~$t^*$ is a leaf in~$H_i$, there is exactly one vertex~$w$ being the neighbor of~$t^*$ in~$H_i$.
  We obtain a terminal Steiner tree~$H$ in~$G$ from~$H_i$ as follows.
  If every edge in~$H_i$ is also present in~$G$, then~$H:=H_i$ also forms a terminal Steiner tree in~$G$.
  Otherwise, there is an inclusion-wise maximal edge set~$E'\subseteq E(H_i)$ such that $E'\cap E(G)=\emptyset$.
  Observe that by construction, the set of endpoints of~$E'$ forms a subset of~$N_G(v_i)$.
  Let initially~$H$ be a copy of~$H_i$. 
  Delete from~$H$ all edges in~$E'$, add vertex~$v_i$ to~$H$, and for each~$\{x,y\}\in E'$, add the edges~$\{x,v_i\}$ and~$\{y,v_i\}$.
  Note that~$H$ remains connected after this step, and the set of leaves remains unchanged.
  Finally, compute a minimum feedback edge set in~$H$ if necessary.
  Observe that since~$V(H)=V(H_i)\cup\{v_i\}$, $H$ forms a terminal Steiner tree of size~$k+|T|$ in~$G$.

  Next, we describe the strict composition for \pprob{Terminal Steiner Tree}{k+|T|}. 
  Given the instances $(G_1,k),\ldots,(G_d,k)$, we create an instance $(G',k)$ as follows.
  Let~$G'$ be initially the disjoint union of~$G_1,\ldots,G_d$.
  For each $t\in T$, identify its copies in~$G_1,\ldots,G_d$, say~$t_1,\ldots,t_d$, with one vertex~$t'$ corresponding to~$t$.
  This finishes the construction of~$G'$.
  Note that for every $i,j\in[d]$, $i\neq j$, any path between a vertex in~$G_i$ and a vertex in $G_j$ contains a terminal vertex.
  Hence, any terminal Steiner tree in $G'$ contains non-terminal vertices only in~$G_i$ for exactly one~$i\in[d]$.
  It is not difficult to see that $(G',k)$ is a yes-instance if and only if one of the instances~$(G_1,k),\ldots,(G_d,k)$ is a yes-instance.
  \end{proof}
}

\paragraph{Multicolored Graph Problems.}
\label{ssec:multicol}
In many cases, a vertex-colored version of the graph problems can help to construct diminishers. As an example, it remains open whether the problem \pprob{Path}{k}, asking whether a given graph contains a simple path of length $k$, is diminishable. However for \pprob{Multicolored Path}{k} we can show that it is.
The \prob{Multicolored Path} problem is defined as follows: Given an undirected graph~$G=(V,E)$ with a vertex coloring function $\col: V \to \{1,\ldots,k\}$, determine whether there exists a simple path of length~$k$ with one vertex of each color. This problem is \NP{}-complete as it generalizes \prob{Hamiltonian Path}. 
Furthermore, \pprob{Multicolored Path}{k} is in FPT as it can be solved in~$2^{O(k)} n^2$~time~\cite{AlonYZ95}.%

\begin{proposition}
\label{lem:multicoloredpathk}
\pprob{Multicolored Path}{k} is diminishable.
\end{proposition}
\begin{proof}
We present a parameter-decreasing branching rule and a strict composition for \pprob{Multicolored Path}{k}. The result will then follow directly from \Cref{lem:branchcomp}. 
Let $(G=(V,E),\col)$ be an instance of the \pprob{Multicolored Path}{k} problem. 
The parameter-decreasing branching rule for~$(G=(V,E),\col)$ creates a graph $G_{(v_1,v_2,v_3)}$ for each ordered triplet $(v_1,v_2,v_3)$ of vertices of~$V$ such that $v_1,v_2,v_3$ is a multicolored path in $G$. 
The graph~$G_{(v_1,v_2,v_3)}$ is constructed from $G$ as follows: We delete from $G$ all vertices $w\in V\setminus\{v_2,v_3\}$ with~$\col(w)\in\{\col(v_1),\col(v_2),\col(v_3)\}$. 
Following this, only vertices of $k-1$ colors remain, and $v_2$ and~$v_3$ are the only vertices colored $\col(v_2)$ and $\col(v_3)$, respectively. 
We then delete all edges incident with~$v_2$, apart from $\{v_2,v_3\}$, and relabel all colors so that the image of $\col$ for~$G_{(v_1,v_2,v_3)}$ is~$\{1,\ldots,k-1\}$.

Clearly our parameter decreasing branching rule can be performed in polynomial time. Furthermore, the parameter decreases in each output instance. 
We show that the first requirement of Definition~\ref{def:branching} holds as well: Indeed, suppose that $G$ has a multicolored path~$v_1,v_2,\ldots,v_k$ of length~$k$. 
Then $v_2,\ldots,v_k$ is a multicolored path of length~$k-1$ in~$G_{(v_1,v_2,v_3)}$ by construction. 
Conversely, suppose there is a multicolored path $u_2,\ldots,u_k$ of length $k-1$ in some $G_{(v_1,v_2,v_3)}$.  
Then since $v_2$ is the only vertex of color $\col(v_2)$ in $G_{(v_1,v_2,v_3)}$, and since $v_2$ is only adjacent to~$v_3$, it must be w.l.o.g.\ that $u_2=v_2$ and $u_3=v_3$. Therefore, since $v_1$ is adjacent to $v_2$ in $G$, and no vertices of $u_2,\ldots,u_k$ have color $\col(v_1)$ in $G$, the sequence of $v_1,u_2,\ldots,u_k$ forms a multicolored path of length $k$ in~$G$.

The strict composition for \pprob{Multicolored Path}{k} is as follows. 
Given a sequence of inputs~$(G_1,\col_1),\ldots,(G_t,\col_t)$, the strict composition constructs the disjoint union $G$ and the coloring function~$\col$ of all graphs $G_i$ and coloring functions $\col_i$, $1\leq i\leq t$. 
Clearly, $(G,\col)$ contains a multicolored path of length $k$ if and only if there is a multicolored path of length $k$ in some~$(G_i,\col_i)$. The result thus follows directly from \Cref{lem:branchcomp}.
\end{proof}

\begin{proposition}%
\label{prop:multicoloredpathklogn}
Unless \PeqNP{}, \pprob{Multicolored Path}{k \log n} has no strict polynomial kernel.
\end{proposition}
{
  \begin{proof}
  Chen, Flum, and M\"{u}ller~\cite[Prop.~3.10]{chen2011lower} proved that if $L$ is a parameterized problem which can be solved in $2^{k^{O(1)}}|x|^{O(1)}$ time, where $k$ is the parameter and $|x|$ is the instance size, then $L(k)$ has a polynomial kernel if and only if $L(k \log |x|)$ has a polynomial kernel. 
  It is easy to verify that their proof also holds for strict polynomial kernels. 
  Thus, as \pprob{Multicolored Path}{k} can be solved in $2^{k^{O(1)}}n^{O(1)}$~\cite{AlonYZ95}, the result follows.
  \end{proof}
}
The idea used in the parameter diminisher for \pprob{Multicolored Path}{k} can also be applied for other problems. The following problem, \prob{Colorful Graph Motif}, asks for a given undirected graph $G=(V,E)$ and a given vertex coloring function $\col:V \to \{1,\ldots,k\}$, whether there exists a connected subgraph of~$G$ containing exactly one vertex of each color. 
\pprob{Colorful Graph Motif}{k} is known to be in \FPT{}~\cite{betzler2008parameterized}.
\begin{proposition}%
\label{lem:coolormotifk}
\pprob{Colorful Graph Motif}{k} is diminishable.
\end{proposition}

{
  \begin{proof}
  We show that there is a parameter-decreasing branching rule and a strict composition for \pprob{Colorful Graph Motif}{k}. Let $(G=(V,E),\col)$ be an instance of \pprob{Colorful Graph Motif}{k}. Assume there are only edges between differently colored vertices. For each~$\{v, w\} \in E$, the parameter-decreasing branching rule creates a graph $G_{\{v, w\}}$ which is a copy of~$G$ where all vertices of colors $\col(v)$ and $\col(w)$ are removed. Furthermore, a new vertex $v^\star$ is added with color $\col(v^\star) = \col(v)$ and with edges to all vertices in~$N(v)\cup N(w)$ that have not been removed. Clearly,~$G_{\{v, w\}}$ only contains vertices of $k-1$ different colors and is computable in polynomial time. Also, if $G$ contains a colorful motif~$\{v_1, v_2, \ldots, v_k\}$ where, without loss of generality, $\{v_1, v_2\} \in E$, then $G_{\{v_1, v_2\}}$ contains the colorful motif~$\{v^\star, v_3, \ldots, v_k\}$. Conversely, if a graph $G_{\{v, w\}}$ contains a colorful motif, then it has to contain~$v^\star$ since it is the only vertex of its color. Let $\{v^\star, v_2, \ldots, v_k\}$ be a colorful motif in $G_{\{v, w\}}$, then $\{v, w, v_2, \ldots, v_k\}$ is a colorful motif in $G$ since $v^\star$ is connected to some vertex $v_i$ in the motif and hence, by construction, $v_i$ is connected to $v$ or to $w$ and there is an edge between $v$ and $w$.

  The strict composition constructs the disjoint union of the sequence of inputs. Clearly, the disjoint union has a colorful motif if and only if one of the input graphs has a colorful motif. 
  \Cref{lem:branchcomp} now yields the result.
  \end{proof}
}

\paragraph{Non-Deterministic Turing Machine Computations.}
\label{ssec:ntms}
Now we turn our attention to single-tape, single-head, non-deterministic Turing machine computations, which also played a significant role in the development of parameterized complexity theory. 
A Turing machine is defined as a tuple $M=(\Sigma, Q, q_0, F, \delta)$, where $\Sigma$ is the alphabet, $Q$ is the set of states, $q_0\in Q$ is the initial state, $F\subseteq Q$ are the accepting states, and~$\delta \subseteq ((\Sigma\cup\{\square\}) \times Q)\times (\Sigma\times Q\times \{-1, 0, 1\})$ is the transition relation, where $\square$ denotes the blank symbol in an empty cell.
The \prob{Short NTM Computation} problem asks, given a Turing machine~$M$, a word $x\in\Sigma^*$ which is initially written on the tape, and a number~$k$ in unary encoding, whether there is a run such that, after $k$ computation steps, the Turing machine $M$ is in an accepting state.
This problem is known to be \W{1}-hard when parameterized by $k$ and in \FPT{} when parameterized by $(k+|\Sigma|)$ and when parameterized by $(k+|Q|)$~\cite{DowneyF99}. In the following we show that the problem does not admit a strict polynomial kernel for those parameterizations unless~\PeqNP{}.

\begin{proposition}%
\label{prop:shortntmc}
\pprob{Short NTM Computation}{k+|\Sigma|} is diminishable.
\end{proposition}
{
  \begin{proof}
  The main idea behind the parameter diminisher for \pprob{Short NTM Computation}{k+|\Sigma|} is to transform the given Turing machine $M$ to a Turing machine $M'=(\Sigma,Q',q'_0,F',\delta')$ that computes the last two steps of $M$, that is step $k-1$ and step $k$, in one step. 
  (If $k=1$, then the parameter diminisher can produce a trivial \YES{}- or \NO-instance.) 
  To do that, we need to encode the letter and its position on the tape that will be read by Turing machine~$M$ in step $k$ in the states of Turing machine $M'$. 
  Additionally, we need to use the states to count the steps in order to allow the Turing machine $M'$ to recognize when it has to compute two steps of $M$ in one step, and we need to encode the position of the tape which the head is currently looking at. Let $x\in \Sigma^*$ denote the input string, that is initially written on the tape, starting at cell $0$.

  Note that in $k$ steps, the head of the Turing machine can potentially only move to and read from cells $-k, -k+1, \ldots, 0, \ldots, k-1, k$, assuming the initial head position is $0$. Hence, we set 
  $$Q' = \{q_0'\}\cup (Q\times (\Sigma\cup \{\square\})\times\{-k,\dots,k\}\times\{-k,\dots,k\}\times\{1,\dots,k-1\})\cup \{\bar{q}\}.$$ 
  So every state different from $q_0'$ and $\bar{q}$ is a tuple consisting of five elements: the state of Turing machine $M$ it corresponds to, the letter that is on a specific position of the tape, that position, the current position of the head, and the current computation step. Every of those states is accepting, it the corresponding state of $M$ is accepting, $q_0'$ is accepting if $q_0$ is accepting, and $\bar{q}$ is not an accepting state.

  For every transition in $\delta$ from the initial state $q_0$ to some state $q$ we create transitions in~$\delta'$ from the new initial state $q_0'$ to all states $(q,x_i, i, m, 1)$ such that the following conditions are met.
  \begin{compactitem}
    \item $-k\leq i\leq k$,
    \item for $1\le i \le |x|$, $x_i$ is the $i$-th letter of input $x$, for $i=0$, $x_0$ is the letter written by $\delta$ to cell~$0$, and $\square$ otherwise, and
    \item $m$ is the movement of the head in the transition $\delta$, that is, $m\in\{-1, 0, 1\}$.
  \end{compactitem}
  The following stays unchanged in the transition: the letter that needs to be read by the head for the transition to be possible, the letter written to the tape, and the movement of the head.
  For every transition from a state $q$ to a state $q'$ in $\delta$, we create transitions from states~$(q, x, i, j, t)$ to states $(q', x', i, j', t+1)$ in $\delta'$, such that the following conditions are met.
  \begin{compactitem}
    \item $-k\leq i\leq k$,
    \item $x'=x$ unless we have that $i=j$, in this case the Turing machine writes symbol $y$ into cell $i$ in that transition $\delta$, and hence $x' = y$,
    \item $j' = j+m$, where $m$ is the movement of the head in the transition $\delta$, that is, $m\in\{-1, 0, 1\}$, and
    \item $1\le t \le k-3$.
  \end{compactitem}
  Again, the following stays unchanged in the transition: the letter that needs to be read by the head for the transition to be possible, the letter written to the tape, and the movement of the head. Finally, for every state $(q, x, i, j, k-2)$ we create a transition in $\delta'$ that on reading symbol $y$ goes to state $(q', x', i, j', k-1)$ without moving the head if the following conditions are met.
  \begin{compactitem}
    \item There is a transition from $q$ on reading symbol $y$ to a state $q^\star$ in $\delta$ with head movement~$m$ such that $j+m=i$, and
    \item there is a transition in $\delta$ that, if $M$ is in state $q^\star$ and reads $x$, goes to state $q'$. 
  \end{compactitem}
  Otherwise, we create a transition in $\delta'$ that goes from state $(q, x, i, j, k-2)$ and reading symbol $y$ to state $\bar{q}$ without moving the head. No writing is necessary.

  Turing machine $M'$ needs to non-deterministically guess the position of the head of Turing machine $M$ in step $k-1$ and then simulates the behavior of $M$ until step $k-2$. If the guess was wrong, then $M$ will transition to $\bar{q}$, a non-accepting state. If the guess was correct, then $M'$ accepts in $k-1$ steps if and only if $M$ accepts in $k$ steps. 
  \end{proof}
}
In \prob{Short Binary NTM Computation} the input Turing machines are restricted to have a two-element alphabet. 
\pprob{Short Binary NTM Computation}{k} is in \FPT{}~\cite{DowneyF99} and it is not hard to see that the parameter diminisher for \pprob{Short NTM Computation}{k+|\Sigma|} is also a parameter diminisher for \pprob{Short Binary NTM Computation}{k}. 
\iflong{}
This yields the following corollary.
  \begin{corollary}
  \pprob{Short Binary NTM Computation}{k} is diminishable.
  \end{corollary}
\fi{}
\begin{proposition}
\pprob{Short NTM Computation}{k+|Q|} is diminishable.
\end{proposition}
\begin{proof}
The idea behind the parameter diminisher for \pprob{Short NTM Computation}{k+|Q|} is to merge the symbols on the tape at positions $-1,0,+1$ into a single new symbol. In this way we can produce a Turing machine $M'$ that computes the first two steps of the given Turing machine~$M$ with a single step. (If $k=1$, then the parameter diminisher can produce a trivial \YES{}- or \NO-instance.)

More specifically, given a Turing machine $M=(\Sigma, Q, q_0, F, \delta)$, a word $x\in\Sigma^*$, and a positive integer $k$, we construct a Turing machine $M'=(\Sigma',Q,q_0,F,\delta')$ and a word $x'\in\Sigma'^*$ such that~$(M',x',k-1)$ is a \YES{}-instance if and only if $(M,x,k)$ is a \YES{}-instance.
We set~$\Sigma'=\Sigma\cup\left((\Sigma\cup\{\square\})\times\Sigma^2\times\{-1,0,+1,S\}\right)$.
The set $\{-1,0,+1,S\}$ encodes the position of the head of $M$ when reading a symbol at positions $-1,0,+1$, while $S$ encodes that this is the first step, and hence that we need to compute two steps of $M$ at once.
The new input $x'$ is set to~$x'=(\square,x_0,x_1,S) x_2 \dots x_n$, where $x_0\dots x_n=x$.

The transition relation $\delta'$ is defined as follows. On symbols in $\Sigma$ the Turing machine~$M'$ behaves like $M$, while on symbols $(\sigma_{-1},\sigma_{0},\sigma_{+1},i)$ with $i\in\{-1, 0, 1\}$, we simulate the original Turing machine $M$ on the character $\sigma_i$. 
If the Turing machine~$M$ moves left or right within~$\{-1,0,1\}$, then we keep the head on the same position and update the value of~$i$. If we move outside the range $\{-1,0,+1\}$, then we move the head to the left and right accordingly. Note that when the Turing machine~$M'$ returns to this symbol, then the value of $i$ will be automatically correct.

Finally, on symbols $(\sigma_{-1},\sigma_{0},\sigma_{+1},S)$ we are in state $q_0$ and simulate the first two steps of the Turing machine $M$. Also, we replace $S$ according to the movement of the head. This guarantees that $M'$ will only read such a symbol in the very first step and hence computes two steps of $M$ in one step exactly once. It is not difficult to see that $M'$ accepts $x'$ in $k-1$ steps if and only if $M$ accepts $x$ in $k$ steps.
\end{proof}

\section{Problems without Semi-Strict Polynomial Kernels}
\label{sec:nodiminishers}
\looseness=-1
\appendixsection{sec:nodiminishers}

Considering \Cref{def:strictPK}, strict kernels only allow an additive increase by a constant of the parameter value.
One may ask whether one can exclude less restrictive versions of strict kernels for parameterized problems using the concept of parameter diminishers.
Targeting this question, in this section we study scenarios with a multiplicative (instead of additive) parameter increase by a constant. We refer to this as \emph{semi-strict kernels}.

\begin{defi}[Semi-Strict Kernel]
\label{def:lessstrictPK}
A \emph{semi-strict kernelization} for a parameterized problem~$L$~is a polynomial-time algorithm that on input instance $(x,k)\in \{0,1\}^* \times \mathbb{N}$ outputs an instance~$(x',k')\in \{0,1\}^* \times \mathbb{N}$, the \emph{semi-strict kernel}, satisfying:
\begin{inparaenum}[(i)]
\item $(x, k) \in L \iff (x', k') \in L$,
\item $|x'| \le f(k)$, for some function $f$, and
\item $k' \le c\cdot k $, for some constant $c$.
\end{inparaenum}
We say that~$L$~admits a semi-strict \emph{polynomial} kernelization if $f(k)\in k^{O(1)}$.
\end{defi}
\iflong{}
  On the one hand, every strict kernelization with constant~$c$ is a semi-strict kernelizations with constant~$c+1$.
  On the other hand, if a parameterized problem~$L$ admits a semi-strict kernel with constant $c$, there is not necessarily a constant $c'$ such that for \emph{every} input instance $(x,k)$, the obtained parameter value~$k'$ of the output instance~$(x',k')$ is upper-bounded by $k'+c'$.
  Hence,~$L$ does not necessarily admit a strict kernelization.
  In this sense, \Cref{def:lessstrictPK} generalizes strict kernelizations.
\fi{}
Note that \Cref{thm:diminisher} in \Cref{sec:framework} does not imply that the problems mentioned in~\Cref{thm:main} do not admit semi-strict polynomial kernelizations unless \PeqNP{}.
Intuitively, in \Cref{thm:diminisher}, the parameter diminisher is constantly often applied to decrease the parameter, while dealing only with a constant additive blow-up of the parameter caused by the strict kernelization.
When dealing with a constant multiplicative blow-up of the parameter caused by the semi-strict kernelization, the parameter diminisher is required to be applied a non-constant number of times.
Hence, to deal with semi-strict kernelization, we introduce a stronger version of our parameter diminisher.

\begin{defi}[Strong Parameter Diminisher]
\label{def:strongdiminisher}
A \emph{strong parameter diminisher} for a parameterized problem $L$ is a polynomial-time algorithm that maps instances $(x,k) \in \{0,1\}^* \times \mathbb{N}$ to instances~$(x', k') \in \{0,1\}^* \times \mathbb{N}$ such that  $(x, k) \in L \text{ if and only if } (x', k') \in L$, and $k' \le k/c$, for some constant $c>1$.
\end{defi}

\begin{remark}
 To simplify arguments in the proofs, we assume without loss of generality that the constant of any strong diminisher is at least two for the remainder of this section.
 Consider a strong parameter diminisher $\calD$ with constant $1<c<2$. 
 Let $\calD'$ be the repetition of $\calD$ exactly $\lceil\log_c 2\rceil$~times.
 Then~$\calD'$ is a strong parameter diminisher with constant $c':=c^{\lceil\log_c 2\rceil}\geq 2$.
\end{remark}
Next, we prove an analogue of \Cref{thm:diminisher} for semi-strict polynomial kernelizations and strong parameter diminishers.

\begin{thm}
\label{thm:strongdiminisher}%
Let $L$ be a parameterized problem such that its unparameterized version is \NP{}-hard and we have that $\{(x,k)\in L\mid k\le c\}\in \PTIME{}$, for some constant $c\geq 1$. %
If $L$ is strongly diminishable and admits a semi-strict polynomial kernel, then \PeqNP{}.
\end{thm}
\begin{proof}
 Let $L$ be a parameterized problem whose unparameterized version is \NP{}-hard and it holds that $\{(x,k)\in L\mid k\le c\}\in \PTIME{}$, for a constant $c\geq 1$. 
 Let $\calD$ be a strong parameter diminisher for $L$ with constant $c_d\geq 2$ and let $\calA$ be a semi-strict polynomial kernelization for~$L$ with constant~$c_a>1$.
 \iflong{}We show that we can solve any instance $(x,k)$ of $L$, with $k$ being the parameter, in polynomial time. \fi{}%
 Let $(x,k)$ be an instance of $L$.
 Apply $\calD$ on $(x,k)$ exactly $c_r:=\lceil\log_{c_d}(c_a+c_d)\rceil$~times to obtain an equivalent instance $\calD^{c_r}(x,k)=(x',k')$ with $k'\leq k/c_d^{c_r}\leq k/(c_a+c_d)$.
 Observe that the size of~$|(x', k')|$ is still polynomial in $|(x, k)|$ as $c_r$ is a constant.
 Next, apply $\calA$ on $(x',k')$ to obtain an equivalent instance $(x'',k'')$ with $|(x'', k'')|\leq k'^{c'}$, $c'\geq 1$, and $k''\leq c_a\cdot k'\leq c_a\cdot k/(c_a+c_d)<k$.
 Repeating the described procedure at most $k$ times produces an instance $(y,1)$ of $L$, solvable in polynomial time. 
\end{proof}
By \Cref{thm:strongdiminisher}, if we can prove a strong diminisher for a parameterized problem, then it does not admit a semi-strict polynomial kernel, unless~\PeqNP{}.
We give a strong diminisher for the \prob{Set Cover} problem: Given a set~$U$ called the universe, a family $\mathcal{F}\subseteq 2^U$ of subsets of $U$, and an integer $k$, the question is whether there are $k$ sets in the family $\mathcal{F}$ that cover the whole universe.
We show that \prob{Set Cover} parameterized by $k\log n$, where $n=|U|$, is strongly diminishable. 

\begin{thm}
 Unless \PeqNP{}, \pprob{Set Cover}{k \log n} and \pprob{Hitting Set}{k \log m} does not admit a semi-strict polynomial kernel.
\end{thm}

\begin{proof}
 Let $(U, \mathcal{F} = \{F_1, \ldots, F_m\}, k)$ be an instance of \pprob{Set Cover}{k\log n} and assume that~$k\geq 2$ and $n\geq 5$.
  If $k$ is odd, then we add a unique element to~$U$, a unique set containing only this element to~$\mathcal{F}$, and we set~$k=k+1$.
  Hence, we assume that~$k$~is even. 
  The following procedure is a strong parameter diminisher for the problem parameterized by $k\log n$. 
  Let $U' = U$ and for all~$F_i, F_j$ create $F'_{\{i, j\}} = F_i \cup F_j$. Let $\mathcal{F}' = \{F'_{\{i, j\}} \ | \ i\neq j\}$ and set $k' = k/2$.
  This yields the instance~$(U', \mathcal{F}', k')$ of \pprob{Set Cover}{k\log n} in polynomial time.
    \iflong{}In the following we \fi{}show that $(U, \mathcal{F}, k)$ is a yes-instance if and only if $(U', \mathcal{F}', k')$ is a yes-instance\iflong{}. Furthermore, we argue that \fi{}$k'\log n' < (k\log n)/c$ for some constant~$c>1$, where $n'=|U'|$.

  \toappendix
  {
    Assume that there is a set cover $\mathcal{C} \subseteq \mathcal{F}$ for~$U$ of size~$k$. 
    Let $\mathcal{C} = \{C_1, C_2, \ldots, C_k\}$ and let $\mathcal{C}' = \{C_1\cup C_2, C_3\cup C_4, \ldots, C_{k-1}\cup C_k\}$. 
    Then clearly $\mathcal{C}'\subseteq \mathcal{F}'$ is a set cover for~$U'$ of size~$k/2$. 
    
    Conversely, assume that there is a set cover $\mathcal{C}'$ for $U'$ of size $k/2$. 
    Let $\mathcal{C}' = \{C_1', C_2', \ldots, C_{k/2}'\}$, let $C_i'=C_i\cup C_{i'}$ and let $\mathcal{C} = \{C_1, C_{1'}, C_2, C_{2'},\ldots,C_{k/2},C_{k/2'}\}$. 
    Then clearly $\mathcal{C}\subseteq \mathcal{F}$ is a set cover for $U$ of size at most~$k$. 
    Furthermore, we have that $m' = {m \choose 2}$. It follows that
    \[k'\log n' \leq \frac{k+1}{2}\log (n+1) \stackrel{k\geq 2}{\leq} \frac{3k}{4}\log (n+1) \leq \frac{\sqrt{3}}{2}k\log ((n+1)^{\sqrt{3}/2}) \stackrel{n\geq 5}{\leq} \frac{\sqrt{3}}{2}k\log (n).\] 
    Note that in the first inequality, we consider the cases of $k$ being modified to be even.
    It follows that for $k>2$ the parameter decreases by at least a factor of $\sqrt{3}/{2}$ and for $k=2$ the parameter diminisher produces either a trivial yes- or a trivial no-instance.
  }
  
  It is easy to see that a strong parameter diminisher for \pprob{Hitting Set}{k \log m} can be constructed in a similar fashion.
\end{proof}

Seeking for strong parameter diminishers to exclude semi-strict polynomial kernelizations raises the question whether there are parameterized problems that are not strongly diminishable.
In the following, we additionally prove that under some complexity-theoretic assumptions, there are natural problems that do not admit strong parameter diminishers. Here we restrict ourselves to problems where we have a regular parameter diminisher.
The complexity-theoretic assumption we base our results on is the \emph{Exponential Time Hypothesis}, or ETH for short~\cite{ImpagliazzoP01,lokshtanov2011lower}.
The Exponential Time Hypothesis states that there is no algorithm for \prob{3-CNF-Sat} running in~$2^{o(n)}$ time, where $n$ denotes the number of variables. 

\begin{thm}
\label{thm:main2}%
Assuming \ETH{}, none of the following problems is strongly diminishable: \pprob{$k$-CNF-Sat}{n}, \pprob{Rooted Path}{k}, \pprob{Clique}{\Delta}, \pprob{Clique}{\tw}, and \pprob{Clique}{\bw}.
\end{thm}
We show the statements of~\Cref{thm:main2} with several
propositions.
The following lemma is the key tool for excluding strong parameter diminishers under \ETH{}.
Roughly, it can be understood as saying that a strong parameter diminisher can improve the running time of existing algorithms.

\begin{lem}%
\label{lem:strdim}
  Let $L$ be a parameterized problem.
  If there is an algorithm $\calA$ that solves any instance~$(x,k)\in L$ in $2^{O(k)}\cdot |x|^{O(1)}$ time and $L$ is strongly diminishable, 
  then there is an algorithm~$\calB$ that solves $L$ in $2^{O(k/ f(x, k))}\cdot |x|^{f(x, k)^{O(1)}}$ time, where $f:L\to \N$ is a function mapping instances of $L$ to the natural numbers with the following property: For every constant $c$ there is a natural number $n$ such that for all instances $(x, k)\in L$ we have that $|x|\ge n$ implies that $f(x, k)\ge c$.
\end{lem}
{
  \newcommand{\cmin}{c}
  \begin{proof}
    Let $L$ be a parameterized problem. 
    Let $\calA$ be an algorithm that solves any instance~$(x,k)\in L$ in~$2^{c_1\cdot k}\cdot |x|^{c_2}$~time with constants $c_1,c_2>0$ and let~$\calD$ be a strong parameter diminisher for $L$ with constant~$d\geq 2$.
    Recall that by definition of a strong parameter diminisher, the size of the instance grows at most polynomially each time $\calD$ is applied. 
    Let $b\geq 1$ be a constant such that the size of the instance obtained by applying $\calD$ once to  $(x,k)$ is upper-bounded by $|x|^b$.
    We set~$\cmin:=\min\{2,b\}$.
    Let~$f:L\to\N$ be a function such that $f(x, k)\ge c$ for all $(x, k)\in L$ with $|x|\ge c_0$ for some constant~$c_0\in\N$. 

    Let $(x',k')$ be the instance of $L$ obtained by applying $\calD$ $\lceil\log_{\cmin} f(x,k)\rceil$ times to instance~$(x,k)\in L$ with $|x|\geq c_0$.
    We obtain 
    \[ |x'| \leq |x|^{b^{\lceil\log_{\cmin} f(x,k)\rceil}} \leq |x|^{b^{2\log_{\cmin} f(x,k)}} \leq |x|^{f(x,k)^{c_3}},\, \text{for some constant $c_3\geq 1$.} \]
    Furthermore, the parameter decreases by the constant factor~$d$ each time the diminisher is applied, hence 
    \[ k' = k/d^{\lceil\log_{\cmin} f(x,k)\rceil}\leq  k/d^{\log_{\cmin} f(x,k)} \leq k/d^{\log_{d} f(x,k)}\leq k/f(x,k) . \] 
    Finally, applying $\calA$ on $(x', k')$ solves $(x',k')$ in time 
    \[2^{c_1\cdot k'}\cdot |x'|^{c_2} \leq  2^{c_1\cdot k/f(x,k)}\cdot |x|^{c_2\cdot f(x,k)^{c_3}} \in 2^{O(k/ f(x,k))}\cdot |x|^{f(x,k)^{O(1)}}.\qedhere\] %
  \end{proof}
}
Intuitively, we apply \Cref{lem:strdim} to exclude the existence of strong parameter diminishers under the above mentioned complexity-theoretic assumptions as follows.
Consider a problem where we know a running time lower bound for any algorithm based on the \ETH{} and we also know an algorithm that matches this lower bound.
Then, due to \Cref{lem:strdim}, for many problems a strong parameter diminisher and a suitable choice for the function $f$ would imply the existence of an algorithm which has a running time that breaks the lower bound.

Chen, Flum, and M\"{u}ller~\cite{chen2011lower} showed that \pprob{$k$-CNF-Sat}{n} and \pprob{Rooted Path}{k} are diminishable. We show that we cannot obtain strong diminishability for these problems unless the \ETH{} breaks. 
  Recall \pprob{$k$-CNF-Sat}{n}, the problem of deciding whether a given Boolean formula with $n$ variables in conjunctive normal form and with at most $k$ literals in each clause is satisfiable.

  \begin{proposition}
  Assuming \ETH{}, \pprob{$k$-CNF-Sat}{n} is not strongly diminishable.
  \end{proposition}
  \begin{proof}
\prob{$k$-CNF-Sat} can be solved in $O^*(2^{n})$ time via a brute-force algorithm~$\calA$, but does not admit a~$2^{o(n)}$ time algorithm under the~\ETH{}.
    By \Cref{lem:strdim} with algorithm~$\calA$ and~$f(\phi)=\log n$, \pprob{$k$-CNF-Sat}{n} does not admit a strong parameter diminisher unless the~\ETH{} breaks.
  \end{proof}

\begin{proposition}
 Assuming \ETH{}, \pprob{Rooted Path}{k} is not strongly diminishable.
\end{proposition}
\begin{proof}
 \prob{Hamiltonian Path} on an $n$-vertex graph reduces trivially to \prob{Rooted Path} by adding a universal vertex and taking it as the root and setting the length of the path $k=n$.
 As \prob{Hamiltonian Path} does not admit a $2^{o(n)}$ time algorithm unless the~\ETH{} breaks~\cite{lokshtanov2011lower}, \prob{Rooted Path} does not admit a~$2^{o(n)}$ time algorithm unless the~\ETH{} breaks.
 There is an algorithm for \pprob{Rooted Path}{k} running in $2^{O(k)}\poly(n)$ time~\cite{AlonYZ95}. Let $(G=(V,E), k)$ be an instance of \pprob{Rooted Path}{k} and set~$f(G, k)=\log(|V|)=\log n$. By \Cref{lem:strdim} we get an algorithm for \pprob{Rooted Path}{k} running in~$2^{O(k/\log n)}\cdot |G|^{(\log n)^{O(1)}}\in 2^{o(n)}$ time.
 Hence, \pprob{Rooted Path}{k} does not admit a strong parameter diminisher unless the~\ETH{} breaks.
\end{proof}

{

  Next, we show that for most parameterizations we considered, \prob{Clique} does not admit a strong parameter diminisher unless the~\ETH{} breaks. 

  \begin{proposition}%
  \label{prop:ethcliquebw}
  Assuming \ETH{}, \pprob{Clique}{\bw} 
  is 
  not strongly diminishable.
  \end{proposition}

  {
    \begin{proof}
      \prob{Clique} can be solved in $O^*(2^{\bw})$ time via a dynamic programming (brute-force) algorithm~$\calA$, but does not admit a $2^{o(n)}$ time algorithm under \ETH{}~\cite{lokshtanov2011lower}. Note that $\bw(G)\in O(n)$.
      By \Cref{lem:strdim} with algorithm $\calA$ and $f(G, k)=\log(|V|)=\log n$, \pprob{Clique}{\bw} does not admit a strong parameter diminisher unless the~\ETH{} breaks. %
    \end{proof}
  }

  Since we have that $\Delta \le \tw \le \bw$, we get the following corollary. Note that we do not obtain this result for \pprob{Clique}{\cw}, since $\cw(G)\in O(n^2)$, where $n$ is the number of vertices of $G$.
  \begin{corollary}
  Assuming \ETH{}, \pprob{Clique}{\Delta} and \pprob{Clique}{\tw} are not strongly diminishable.
  \end{corollary}

}

\iflong{}%
Finally, note that it is not hard to observe that 
if we can exclude a strong parameter diminisher for a problem~$L$ parameterized by~$k$ under ETH, 
then we can exclude a parameter diminisher for~$L$ parameterized by~$\log k$ under~ETH.
Thus, it would be interesting to know whether there is a way to 
exclude the existence of parameter diminishers avoiding this exponential gap between the 
parameterizations.
\fi{}

\section{Conclusion}
\looseness=-1
\label{sec:conclusion}

We showed that for several natural problems a %
strict polynomial-size problem
kernel is as likely as~\PeqNP{}. 
Since basically all observed 
(natural and practically relevant) polynomial kernels are 
strict, this reveals 
that the existence of valuable kernels may be tighter connected to 
the P vs.\ NP problem than previously 
expected (in almost all previous work a 
connection is drawn to a collapse of the 
polynomial hierarchy to its third level, and the conceptual framework used
there seems more technical than the one used here).
Our work has been triggered by results of 
Chen, Flum, and M\"uller~\cite{chen2011lower} and shows %
that their basic ideas 
can be extended to a larger class of problems than dealt with 
in their work.

Our work leaves several challenges for future work.
Of course, it would be desirable to find natural problems, where the presented framework is able to refute strict polynomial kernels while the framework of Bodlaender et al.~\cite{BodlaenderDFH09} is not. 
However, it is not clear whether a framework based on a weaker assumption is even able to produce results that a framework based on a stronger assumption is not able to produce. 
This possibly also ties in with the question whether there are parameterized problems that admit a polynomial kernel but no strict polynomial kernel. 
In the following, we list some concrete open problems:
\begin{compactitem}
\item We proved that \pprob{Multicolored Path}{k} is diminishable (and thus 
does not admit a strict polynomial kernel unless \PeqNP{}). Can this
result be extended to the uncolored version of the problem?
This is also open for the directed case.
\iflong{}
  \item Is \pprob{Connected Vertex Cover}{k} diminishable?
\fi{}
  \item Is \pprob{Internal Steiner Tree}{k+|T|}~\cite{HuangLGH13} diminishable? 
  \item \pprob{Clique}{\Delta}, \pprob{Clique}{\tw}, 
\pprob{Clique}{\bw} do not have strong diminishers under the \ETH{} (\Cref{sec:nodiminishers}). 
Is this also true for \pprob{Clique}{\cw}? 
\end{compactitem}

\bibliographystyle{plainurl}
\bibliography{diminisher}

\begin{thebibliography}{10}

\bibitem{AbuKhzamF06}
Faisal~N. Abu{-}Khzam and Henning Fernau.
\newblock Kernels: {A}nnotated, proper and induced.
\newblock In {\em Proc.~2nd {IWPEC}}, volume 4169 of {\em LNCS}, pages
  264--275. Springer, 2006.

\bibitem{AlonYZ95}
Noga Alon, Raphael Yuster, and Uri Zwick.
\newblock Color-coding.
\newblock {\em J. {ACM}}, 42(4):844--856, 1995.

\bibitem{betzler2008parameterized}
Nadja Betzler, Michael~R. Fellows, Christian Komusiewicz, and Rolf Niedermeier.
\newblock Parameterized algorithms and hardness results for some graph motif
  problems.
\newblock In {\em Proc.\ 19th {CPM}}, volume 5029 of {\em LNCS}, pages 31--43.
  Springer, 2008.

\bibitem{BGKN11}
Nadja Betzler, Jiong Guo, Christian Komusiewicz, and Rolf Niedermeier.
\newblock Average parameterization and partial kernelization for computing
  medians.
\newblock {\em J.~Comput.~Syst.~Sci.}, 77(4):774--789, 2011.

\bibitem{BiniazMS15}
Ahmad Biniaz, Anil Maheshwari, and Michiel H.~M. Smid.
\newblock On the hardness of full {S}teiner tree problems.
\newblock {\em J. Discrete Algorithms}, 34:118--127, 2015.

\bibitem{BFFLSV12}
Daniel Binkele{-}Raible, Henning Fernau, Fedor~V. Fomin, Daniel Lokshtanov,
  Saket Saurabh, and Yngve Villanger.
\newblock Kernel(s) for problems with no kernel: On out-trees with many leaves.
\newblock {\em {ACM} Trans.~Algorithms}, 8(4):38, 2012.

\bibitem{BodlaenderDFH09}
Hans~L. Bodlaender, Rodney~G. Downey, Michael~R. Fellows, and Danny Hermelin.
\newblock On problems without polynomial kernels.
\newblock {\em J. Comput. Syst. Sci.}, 75(8):423--434, 2009.

\bibitem{CaiCDF97}
Liming Cai, Jianer Chen, Rodney~G. Downey, and Michael~R. Fellows.
\newblock Advice classes of parameterized tractability.
\newblock {\em Ann.~Pure Appl.~Logic}, 84(1):119--138, 1997.

\bibitem{CFKX07}
Jianer Chen, Henning Fernau, Iyad~A. Kanj, and Ge~Xia.
\newblock Parametric duality and kernelization: Lower bounds and upper bounds
  on kernel size.
\newblock {\em {SIAM} J.~Comput.}, 37(4):1077--1106, 2007.

\bibitem{ChenKJ01}
Jianer Chen, Iyad~A. Kanj, and Weijia Jia.
\newblock Vertex cover: Further observations and further improvements.
\newblock {\em J. Algorithms}, 41(2):280--301, 2001.

\bibitem{chen2011lower}
Yijia Chen, J{\"o}rg Flum, and Moritz M{\"u}ller.
\newblock Lower bounds for kernelizations and other preprocessing procedures.
\newblock {\em Theory Comput.~Syst.}, 48(4):803--839, 2011.

\bibitem{cygan2015parameterized}
Marek Cygan, Fedor~V. Fomin, {\L}ukasz Kowalik, Daniel Lokshtanov, D{\'a}niel
  Marx, Marcin Pilipczuk, Micha{\l} Pilipczuk, and Saket Saurabh.
\newblock {\em Parameterized Algorithms}.
\newblock Springer, 2015.

\bibitem{CyganPPW12}
Marek Cygan, Marcin Pilipczuk, Michal Pilipczuk, and Jakub~Onufry Wojtaszczyk.
\newblock Kernelization hardness of connectivity problems in $d$-degenerate
  graphs.
\newblock {\em Discrete Appl. Math.}, 160(15):2131--2141, 2012.

\bibitem{Diestel10}
Reinhard Diestel.
\newblock {\em Graph Theory}.
\newblock Springer, 4th edition, 2010.

\bibitem{DLS14}
Michael Dom, Daniel Lokshtanov, and Saket Saurabh.
\newblock Kernelization lower bounds through colors and {ID}s.
\newblock {\em {ACM} Trans.~Algorithms}, 11(2):13:1--13:20, 2014.

\bibitem{DowneyF99}
Rodney~G. Downey and Michael~R. Fellows.
\newblock {\em Parameterized Complexity}.
\newblock Monographs in Computer Science. Springer, 1999.

\bibitem{DowneyF13}
Rodney~G. Downey and Michael~R. Fellows.
\newblock {\em Fundamentals of Parameterized Complexity}.
\newblock Springer, 2013.

\bibitem{DreyfusW71}
S.~E. Dreyfus and R.~A. Wagner.
\newblock The {S}teiner problem in graphs.
\newblock {\em Networks}, 1(3):195--207, 1971.

\bibitem{FKRS12}
Michael~R. Fellows, Ariel Kulik, Frances~A. Rosamond, and Hadas Shachnai.
\newblock Parameterized approximation via fidelity preserving transformations.
\newblock In {\em Proc.\ 39th {ICALP}}, volume 7391 of {\em LNCS}, pages
  351--362. Springer, 2012.

\bibitem{FlumG06}
J{\"{o}}rg Flum and Martin Grohe.
\newblock {\em Parameterized Complexity Theory}.
\newblock Springer, 2006.

\bibitem{FortnowS11}
Lance Fortnow and Rahul Santhanam.
\newblock Infeasibility of instance compression and succinct {PCP}s for {NP}.
\newblock {\em J. Comput. Syst. Sci.}, 77(1):91--106, 2011.

\bibitem{GN07}
Jiong Guo and Rolf Niedermeier.
\newblock Invitation to data reduction and problem kernelization.
\newblock {\em ACM {SIGACT} News}, 38(1):31--45, 2007.

\bibitem{HermelinKSWW15}
Danny Hermelin, Stefan Kratsch, Karolina Soltys, Magnus Wahlstr{\"{o}}m, and
  Xi~Wu.
\newblock A completeness theory for polynomial ({T}uring) kernelization.
\newblock {\em Algorithmica}, 71(3):702--730, 2015.

\bibitem{HuangLGH13}
Chao{-}Wen Huang, Chia{-}Wei Lee, Huang{-}Ming Gao, and Sun{-}Yuan Hsieh.
\newblock The internal {S}teiner tree problem: {H}ardness and approximations.
\newblock {\em J. Complexity}, 29(1):27--43, 2013.

\bibitem{ImpagliazzoP01}
Russell Impagliazzo and Ramamohan Paturi.
\newblock On the complexity of $k$-{SAT}.
\newblock {\em J. Comput. Syst. Sci.}, 62(2):367--375, 2001.

\bibitem{karp1982turing}
Richard~M. Karp and Richard Lipton.
\newblock Turing machines that take advice.
\newblock {\em L'{E}nseignement {M}ath{\'e}matique}, 28(2):191--209, 1982.

\bibitem{Kra14}
Stefan Kratsch.
\newblock Recent developments in kernelization: {A} survey.
\newblock {\em Bulletin of the {EATCS}}, 113, 2014.

\bibitem{LinX02}
Guo{-}Hui Lin and Guoliang Xue.
\newblock On the terminal {S}teiner tree problem.
\newblock {\em Inf. Process. Lett.}, 84(2):103--107, 2002.

\bibitem{LinFCL17}
Mugang Lin, Qilong Feng, Jianer Chen, and Wenjun Li.
\newblock Partition on trees with supply and demand: Kernelization and
  algorithms.
\newblock {\em Theor. Comput. Sci.}, 657:11--19, 2017.

\bibitem{lokshtanov2011lower}
Daniel Lokshtanov, D{\'a}niel Marx, and Saket Saurabh.
\newblock Lower bounds based on the {E}xponential {T}ime {H}ypothesis.
\newblock {\em Bulletin of the EATCS}, (105):41--72, 2011.

\bibitem{LokshtanovPRS17}
Daniel Lokshtanov, Fahad Panolan, M.~S. Ramanujan, and Saket Saurabh.
\newblock Lossy kernelization.
\newblock In {\em Proc.\ 49th {STOC}}, pages 224--237. {ACM}, 2017.

\bibitem{Niedermeier06}
Rolf Niedermeier.
\newblock {\em Invitation to Fixed-Parameter Algorithms}.
\newblock Oxford University Press, 2006.

\bibitem{SKMN12}
Alexander Sch{\"{a}}fer, Christian Komusiewicz, Hannes Moser, and Rolf
  Niedermeier.
\newblock Parameterized computational complexity of finding small-diameter
  subgraphs.
\newblock {\em Optimization Letters}, 6(5):883--891, 2012.

\bibitem{west}
Douglas~B. West.
\newblock {\em Introduction to Graph Theory}.
\newblock Prentice Hall, 2 edition, 2000.

\bibitem{XiaoK17}
Mingyu Xiao and Shaowei Kou.
\newblock Kernelization and parameterized algorithms for 3-path vertex cover.
\newblock In {\em Proc.\ 14th {TAMC}}, volume 10185 of {\em LNCS}, pages
  654--668, 2017.

\end{thebibliography}

\newpage
\appendix

\section{Problem Zoo}
\label{sec:probzoo}
\decprob{Biclique}
{An undirected bipartite graph $G=(V=A\uplus B,E)$ and an integer~$k$.}
{Is there a vertex set $X\subseteq V$ of $G$ such that $|X\cap A|=|X\cap B|= k$ and all vertices in $X\cap A$ are adjacent with all vertices in $X\cap B$?}

\decprob{Clique}
{An undirected graph $G=(V,E)$ and an integer~$k$.}
{Is there a vertex set $X\subseteq V$ of $G$ such that $|X|\geq k$ and all vertices in $X$ are pairwise adjacent in $G$?}

\decprob{CNF Sat}
{Given a Boolean formula~$\phi$ in conjunctive normal form (CNF).}
{Is $\phi$ satisfiable?}

\decprob{Colorful Graph Motif}
{An undirected graph $G=(V,E)$, an integer $k$, and a vertex coloring function $\col:V \to \{1,\ldots,k\}$.}
{Is there a vertex set~$X\subseteq V$ such that $G[X]$ is connected and $X$ contains exactly one vertex of each color?}

\decprob{Connected Vertex Cover}
{An undirected graph $G=(V,E)$ and an integer~$k$.}
{Is there a vertex set $X\subseteq V$ in $G$ with $|X|\le k$ and each edge of $G$ is incident to at least one vertex in $X$ and $G[X]$ is connected?}

\decprob{Directed Path}
{A directed graph~$G=(V,E)$ and an integer~$k$.}
{Is there a simple directed path $P$ of length at least~$k$ in~$G$?}

\decprob{Hamiltonian Path}
{An undirected graph~$G=(V,E)$.}
{Is there a path in $G$ that visits each vertex exactly once?}

\decprob{Hitting Set}
{Given a universe $U$, a family $\mathcal{F}\subseteq 2^U$ of subsets of $U$, and an integer~$k$.}
{Is there a subset $U'\subseteq U$ such that $|U'|\leq k$ and $F\cap U'\neq \emptyset$ for all $F\in\mathcal{F}$?}

\decprob{Independent Set}
{An undirected graph $G=(V,E)$ and an integer~$k$.}
{Is there a vertex set $X\subseteq V$ of $G$ such that $|X|\geq k$ and $G[X]$ is edge-free?}

\decprob{Internal Steiner Tree}
{An undirected graph $G=(V=N\uplus T,E)$ and an integer~$k$.}
{Is there a subgraph~$H$ of~$G$ such that $H$ is a tree with~$T$ being part of its internal vertices?}

\decprob{$k$-CNF Sat}
{Given a Boolean formula~$\phi$ in conjunctive normal form (CNF) with at most $k$ literals in each clause.}
{Is $\phi$ satisfiable?}

\decprob{Multicolored Path}
{An undirected graph~$G=(V,E)$ and a vertex coloring function $\col:V \to \{1,\ldots,k\}$.}
{Is there a simple path $P$ in $G$ that contains exactly one vertex of each color?}

\decprob{Path}
{An undirected graph~$G=(V,E)$ and an integer~$k$.}
{Is there a simple path $P$ of length at least~$k$ in~$G$?}

\decprob{Rooted Path}
{An undirected graph~$G=(V,E)$, an integer~$k$, and a root vertex $r\in V$.}
{Is there a simple path $P$ of length at least~$k$ in~$G$ starting at root vertex~$r$?}

\decprob{Set Cover}
{Given a universe $U$, a family of sets $\mathcal{F}\subseteq 2^U$, and an integer~$k$.}
{Is there a subset $\mathcal{C}\subseteq \mathcal{F}$ such that $|\mathcal{C}|\leq k$ and $U=\bigcup_{C\in \mathcal{C}} C$?}

\decprob{Short Binary NTM Computation}
{Given a single-tape, single-head, non-deterministic Turing machine $M=(\Sigma, Q, q_0, F, \delta)$, where $\Sigma$ is an alphabet of size two, that is $|\Sigma| = 2$, $Q$ is the set of states, $q_0\in Q$ is the initial state, $F\subseteq Q$ are the accepting states, and $\delta \subseteq ((\Sigma\cup\{\square\}) \times Q)\times (\Sigma\times Q\times \{-1, 0, 1\})$ is the transition relation, where $\square$ denotes the blank symbol in an empty cell, a word $x\in\Sigma^*$, and an integer $k$ unary encoding.}
{Does $M$ reach an accepting state after $k$ computation steps if the word initially written on the tape is $x$?}

\decprob{Short NTM Computation}
{Given a single-tape, single-head, non-deterministic Turing machine $M=(\Sigma, Q, q_0, F, \delta)$, where $\Sigma$ is the alphabet, $Q$ is the set of states, $q_0\in Q$ is the initial state, $F\subseteq Q$ are the accepting states, and $\delta \subseteq ((\Sigma\cup\{\square\}) \times Q)\times (\Sigma\times Q\times \{-1, 0, 1\})$ is the transition relation, where $\square$ denotes the blank symbol in an empty cell, a word $x\in\Sigma^*$, and an integer $k$ unary encoding.}
{Does $M$ reach an accepting state after $k$ computation steps if the word initially written on the tape is $x$?}

\decprob{Terminal Steiner Tree}
{An undirected graph $G=(V=N\uplus T,E)$ and an integer~$k$.}
{Is there a subgraph~$H$ of~$G$ such that $H$ is a tree with~$T$ being its set of leaves?}

\decprob{Vertex Cover}
{An undirected graph $G=(V,E)$ and an integer~$k$.}
{Is there a vertex set $X\subseteq V$ in $G$ with $|X|\le k$ and each edge of $G$ is incident to at least one vertex in $X$?}

\end{document}